\documentclass[12pt]{amsart}

\usepackage{graphicx} 
\usepackage{amsmath, amssymb, mathtools, mathrsfs}
\usepackage{tikz-cd}
\usetikzlibrary{arrows, arrows.meta}
\usepackage{float}

\newcommand{\C}{\mathbb C}
\newcommand{\Z}{\mathbb Z}
\newcommand{\cH}{\mathcal{H}}
\DeclareMathOperator{\Res}{Res}
\DeclareMathOperator{\Span}{Span}

\numberwithin{equation}{section}
\newtheorem{theorem}{Theorem}[section]
\newtheorem{lemma}[theorem]{Lemma}
\newtheorem{prop}[theorem]{Proposition}

\newtheorem{example}[theorem]{Example}
\theoremstyle{definition}

\newtheorem{remark}[theorem]{Remark}

\usepackage{tikz}
\usepackage{tikz-cd}
\usepackage{color}
\usepackage{hyperref}
\usepackage{longtable}

\usepackage{graphicx}
\graphicspath{ {./images/} }
\allowdisplaybreaks

\title[Generalised spin Calogero--Moser systems]{Generalised spin Calogero--Moser systems from Cherednik algebras}
\author{Misha Feigin, Mikhail Vasilev, Martin Vrabec}
\address{School of Mathematics and Statistics, University of Glasgow, University Place, Glasgow, G12 8QQ, United Kingdom}
\email{Misha.Feigin@glasgow.ac.uk, Mikhail.Vasilev@glasgow.ac.uk}
\address{Université de Montréal, Centre de Recherches Mathématiques, Pavillon André-Aisenstadt, 2920 Chemin de la tour, Montréal, H3T 1J4, Canada}
\email{martinvrabec222@gmail.com}

\begin{document}

\begin{abstract}
    Integrable spin Calogero--Moser type systems with non-symmetric configurations
    of the singularities of the potential appeared in the work of Chalykh, Goncharenko, 
    and Veselov in 1999. We obtain
        various generalisations of these examples 
    by making use of the representation theory of Cherednik algebras.
\end{abstract}

\maketitle

\tableofcontents


\section{Introduction}

The usual Calogero--Moser integrable model describes a system of pairwise interacting particles on a line (rational case) or a circle (trigonometric case) with an inverse square distance potential~\cite{Calogero, Moser, Sutherland}. This potential is naturally associated with the root system of type $A$. Integrable generalisations for other root systems were introduced by Olshanetsky and Perelomov~\cite{OP}.

  Quantum integrable generalisations for some special configurations of vectors that are not root systems were first discovered by Chalykh, Veselov, and one of the authors~\cite{VFC, CFV'98, CFV'99}. Namely, they found the configurations $A_n(m)$ and $C_{n+1}(m,l)$, which are certain deformations of the root systems $A_n$ and $C_{n+1}$, respectively.   
  A different deformation of the $A$-type root system appeared in~\cite{CVlocus}.
  The configuration $A_n(m)$ leads to a Calogero--Moser type system where one of the particles has a different mass. In \cite{SV}, Sergeev and Veselov considered a generalisation of this system to an arbitrary number of particles of each of the two types. They also introduced several other configurations, including a deformation of the root system $BC_n$ that can be thought of as its generalisation to the case of two types of particles. The systems with two types of particles were also considered in the context of gauge theory \cite{Kimura}. Some systems containing three different types of particles were studied recently in \cite{GR, BC}.
  Further configurations leading to quantum integrable Calogero--Moser type systems have been discovered and investigated in~\cite{Feigin, FF, FV19}. 
  In the rational case in dimension $2$, there are many integrable cases. They were investigated by Berest and Loutsenko in~\cite{BL} in relation to Hadamard's problem  on hyperbolic differential equations satisfying Huygens' principle in a strong sense.

  
  For the study of special polynomial eigenfunctions of the systems with two types of particles,
  we refer the reader to \cite{Sergeev, SV'05, SV'09, HL, AHL, Feigin1, Hallnas}. In particular, the Jack polynomials are eigenfunctions of the trigonometric Calogero--Moser operator of type~$A$, and its deformation associated with the configuration $A_n(m)$ has super-Jack polynomials as eigenfunctions (see~\cite{Sergeev}).

A uniform proof of the integrability of the Olshanetsky--Perelomov operators using Dunkl operators~\cite{Dunkl} was obtained by Heckman~\cite{Heckman, HeckmanTrig}, which led to a connection between Cherednik algebras and quantum Calogero--Moser systems associated with root systems (see \cite{Et}). Rational generalised Calogero--Moser systems for non-symmetric configurations of vectors with special coupling parameters were related to the representation theory of rational Cherednik algebras by one of the authors in~\cite{Feigin}. A similar approach to the elliptic case is developed in~\cite{FS}.
Another approach to rational generalised Calogero--Moser systems using Cherednik algebras was proposed by Berest and Chalykh in~\cite{BC}. 

A quantum matrix (spin) version of the Calogero--Moser model appeared in the work of Ha and Haldane \cite{HH}, where they considered matrix differential operators acting on functions valued in the $N$-fold tensor product of vector spaces. Its integrability was established by Minahan and Polychronakos, 
who constructed commuting quantum integrals 
using Polychronakos' version of Dunkl operators in \cite{MP}, and independently at the same time by Hikami and Wadati in \cite{HW} by introducing a Lax pair. 
Bernard, Gaudin, Haldane, and Pasquier produced extra quantum integrals for this model in~\cite{BAB} by using a Yangian symmetry. 

For the Olshanetsky--Perelomov operators, the quantum matrix version 
was considered by Cherednik in~\cite{Ch2}. In this case the Hamiltonian acts on functions of coordinates taking values in a module over the corresponding Weyl group (or Coxeter group in the rational case). Quantum integrals can be obtained from suitable restrictions of invariant combinations of Dunkl operators (see formula (2.30) for the Hamiltonian and Corollary 2.8 in \cite{Ch2} for the commutative family).

Other versions of matrix Calogero--Moser systems are known as well. The first appearance goes back to the work of Gibbons and Hermsen~\cite{GH}, in the classical case. Relation of their system with the Yangian was found in \cite{BAB1}, and this system was related to the pole dynamics of solutions of the matrix KP equation in~\cite{KBBT} (see \cite{Krich} for a relation of the scalar Calogero--Moser dynamics with  solutions of the KP equation). A type $B$ spin version of this system appeared in \cite{ChS} (together with corresponding nonlinear PDEs and generalisations for other wreath product complex reflection groups). It was shown in \cite{FH} that an addition of the harmonic potential preserved (super)integrability of these models.
For further spin Calogero--Moser type systems we refer to \cite{LX} and references therein. In particular, matrix Calogero--Moser type systems were associated to coadjoint orbits in the dual of Lie algebras in \cite{Resh}, and their superintegrability together with action-angle duality was studied. Recently a blending of classical and quantum spin Calogero--Moser type systems appeared in \cite{LRS}. Let us also refer to papers \cite{Crampe} and \cite{CrampeYoung} dealing with other matrix Calogero--Moser type systems related to Lie algebras and complex reflection groups, respectively.

First quantum matrix models associated with configurations of vectors that are not root systems were introduced in the rational case by Chalykh, Goncharenko, and Veselov in~\cite{CGV}, where they proved the $D$-integrability of these systems. They considered matrix operators of the form
\begin{equation*}
    L = \Delta - \sum\limits_{\alpha \in \mathcal{A}}\frac{m_{\alpha}(m_{\alpha} - P_{\alpha}) (\alpha, \alpha)}{(\alpha, x)^2},
\end{equation*} 
where $\mathcal{A}$ is a finite collection of vectors in a vector space $V$, $P_{\alpha}$ is a matrix acting on the vector space $U \cong V$ as a reflection about the hyperplane orthogonal to $\alpha$, and $m_\alpha$ is an integer scalar multiplicity.
In addition to root systems, where one recovers systems from \cite{Ch2}, one of the examples of such a configuration $\mathcal A$ considered in~\cite{CGV} is the deformed $A_n$-type collection of vectors $A_{n}(m)$ consisting of the vectors $\{e_i - e_j \mid 1 \leq i < j \leq n \}$ with multiplicity~$m_{e_i -e_j} = m$ and $\{e_i - \sqrt{m} e_{n + 1} \mid 1 \leq i \leq n \}$ with multiplicity $m_{e_i - \sqrt{m} e_{n + 1}} = 1$. 
Another example considered in~\cite{CGV} is the deformed $C$-type configuration $C_{n+1}(m,l)$.

In this paper, we obtain 
generalisations of these operators $L$ admitting many 
commuting integrals, to various other configurations of vectors $\mathcal{A}$, to more general vector spaces $U$, and to the trigonometric case by using the representation theory of Cherednik algebras.

We start with the polynomial representation $\mathcal{P}$ of the rational or trigonometric Cherednik algebra, which is realised 
with the help of Dunkl operators. For special values of the parameters of the Cherednik algebra, the polynomial representation becomes reducible with a submodule $\mathcal{I}$ given by polynomials vanishing on the orbit of an intersection of Coxeter mirrors. We enlarge  
the quotient representation~$\mathcal{P}/\mathcal{I}$ of the Cherednik algebra by taking the tensor product with 
a right module $\mathcal{U}$ of the corresponding Coxeter group $W$. The action of invariant polynomials of Dunkl operators on the diagonal invariants $(\mathcal{U} \otimes \mathcal{P}/\mathcal{I})^{W}$ produces commuting differential operators with values in endomorphisms of $U = \mathcal{U}^{W_0}$, where $W_0 \subset W$ is the parabolic subgroup associated with the submodule $\mathcal{I}$. This extends the construction of scalar generalised Calogero--Moser systems from~\cite{Feigin}, which can be recovered by setting~$\mathcal{U}$ to be the trivial representation.

The structure of the paper is as follows. In Section~\ref{invariant ideals}, we review the definition of parabolic strata for finite Coxeter groups and the conditions that the strata have to satisfy for the ideal of polynomials vanishing on them to be invariant under the associated rational Cherednik algebra. In Section~\ref{rational case}, we show that parabolic strata defining invariant ideals can be used to obtain matrix operators of rational spin Calogero--Moser type, including the case with a harmonic term, and quantum integrals for them, and explain the relation of these Hamiltonians to projections of Coxeter root systems. We compute the Hamiltonians explicitly for all classical root systems. 

In Section~\ref{trig case}, we develop an analogous construction for trigonometric Cherednik algebras associated with crystallographic root systems (including the non-reduced root system of type $BC$), which leads to generalisations of trigonometric spin Calogero--Moser operators. 
We present the resulting Hamiltonians explicitly in the case of all classical root systems and also in several interesting examples involving the exceptional types. For type $A$, we then provide in Section~\ref{extra integals} also additional quantum integrals for the corresponding deformed trigonometric spin Calogero--Moser operators by exploiting a Yangian symmetry.

\section{Invariant parabolic ideals for rational Cherednik algebras}\label{invariant ideals}
Here we review the basic information about invariant parabolic ideals for the polynomial
representation of rational Cherednik algebras following \cite{Feigin}. 

Let $W$ be a finite real
reflection group acting by orthogonal transformations in the complexified reflection 
representation $V = \mathbb{C}^N$. Let~$R$ and $\Gamma$ be the corresponding (reduced) root system and
Coxeter graph, respectively. We assume that a positive subsystem $R_{+} \subset R$ is 
chosen so that vertices of $\Gamma$ are identified with simple roots.
We denote by $\Gamma_0 \subset \Gamma$ and $\Gamma_0^v$ a subgraph of the Coxeter graph and the 
set of simple roots corresponding to the vertices of the subgraph $\Gamma_0$, respectively. We denote by $W_0 \subset W$ the parabolic subgroup of the Coxeter group generated by the reflections about the hyperplanes orthogonal to the roots $\Gamma_0^v$.

Let $c \colon R \to \C$, $\alpha \mapsto c_{\alpha}$ be a $W$-invariant function.  
The rational Cherednik
algebra $\cH_c$ associated with the root system $R$ is the algebra acting faithfully on the space of polynomials
$\mathbb{C}[x] = \mathbb{C}[x_1, \ldots, x_N]$ generated by polynomials $p \in \mathbb{C}[x]$, the reflection group~$W$, and Dunkl operators
\begin{equation}
    \label{Dunkl}
    \nabla_{\xi} = \partial_{\xi} - \sum\limits_{\alpha \in R_{+}} \frac{c_{\alpha} (\alpha, \xi)}{(\alpha, x)}(1 - s_{\alpha}), \qquad \quad (\xi \in V)
\end{equation}
where $(\cdot, \cdot)$ is the standard $\C$-bilinear form on $V$, $\partial_\xi$ is the directional derivative along the vector $\xi$, and $s_\alpha$ is the orthogonal reflection about the hyperplane $(\alpha, x)$ = 0 \cite{EG}.
It is well known that Dunkl operators commute among themselves, $[\nabla_{\xi} , \nabla_{\eta}] = 0$ for all~$\xi, \eta \in V$ \cite{Dunkl}. Let $\{e_i\}_{i=1}^N$ be the standard orthonormal basis of $V$, and let $\nabla_i = \nabla_{e_i}$.

Let $\Gamma_0$ be a subgraph of the Coxeter graph $\Gamma$ obtained by
specifying some of the vertices of $\Gamma$ and preserving all the edges between these vertices.
The subgraph $\Gamma_0$ defines the subspace
\begin{equation*}
    \pi = \pi_{\Gamma_0} = \{x \in V \mid (\beta , x)=0, \ \forall \beta \in \Gamma_0^v \}.
\end{equation*}
The associated parabolic stratum is defined as 
\begin{equation*}
    D_{\Gamma_0} = \bigcup\limits_{w \in W} w(\pi).
\end{equation*}
The corresponding parabolic ideal $I_{\Gamma_0}$ is the set of polynomials vanishing on the stratum,
$I_{\Gamma_0} = \{p \in \mathbb{C}[x] \mid p|_{D_{\Gamma_0}} = 0 \}$. Then the following theorem
holds.

\begin{theorem} \cite{Feigin}
    \label{invariance}
    Let $\Gamma_0 = \coprod_{i = 1}^l \Gamma_i$ be the decomposition of the subgraph into connected components. Then the parabolic ideal $I_{\Gamma_0}$ is invariant under the rational Cherednik algebra $\cH_c$ if and only if the following relation holds for all $i=1, \dots, l$:
    \begin{equation*}
        \sum\limits_{\alpha \in R \cap V_i} \frac{c_{\alpha} (\alpha , u) (\alpha, v)}{(\alpha, \alpha)} = (u, v)
    \end{equation*}
    for all $u,v \in V_i$, where $V_i$ is the linear space spanned by the roots $\Gamma_i^v$.
\end{theorem}

\section{Generalised rational spin Calogero--Moser systems}\label{rational case} 

In this section, we explain how $\cH_c$-invariant parabolic ideals $I_{\Gamma_0}$ lead to operators of rational spin (equivalently, matrix) Calogero--Moser type and quantum integrals for them.

\subsection{Review of the scalar case}

Suppose $I_{\Gamma_0}$ is an invariant parabolic ideal as above. Let us define the quotient module 
$M = \mathbb{C}[x]/I_{\Gamma_0}$. Any polynomial function of Dunkl operators can act on the module~$M$. Moreover, $W$-invariant polynomials of Dunkl operators preserve the space of invariants $M^{W}$. Consider the operators
\begin{equation*}
    H_p = \Res p(\nabla)
\end{equation*}
for $W$-invariant polynomials $p \in \mathbb{C}[x]^W$, $p(\nabla) = p(\nabla_1, \ldots, \nabla_N)$, where $\Res$ denotes restriction to $M^W$. Let us note that $M^W$ is, more generally, preserved by
the centraliser $\cH_c^W=C_{\cH_c}(\mathbb{C}[W])$ of the group algebra~$\mathbb{C}[W]$ of the Coxeter group inside the rational Cherednik algebra. In particular, it is a module for the spherical subalgebra $S\cH_c \subset \cH_c^W$, where $S\cH_c = e\cH_ce$ with $e = |W|^{-1} \sum_{w \in W}w$. 

To write down the operators $H_p$ explicitly in local coordinates on $\pi$, it is convenient to consider action of $p(\nabla)$ on $W$-invariant (formal) sums of germs of analytic functions on $W$-orbits of small neighbourhoods on~$D_{\Gamma_0}$ of a generic point of $\pi$  (see~\cite[Section~3]{Feigin} for the details), rather than on global functions. This way the operator $H_p$ takes the form of a differential operator on $\pi$, which we denote as $\Res_\pi p(\nabla)$. 
In particular, by using that
\begin{equation}\label{q1}
    \sum\limits_{i = 1}^N \nabla_i^2  = \sum\limits_{i = 1}^N \partial_{x_i}^2 - \sum\limits_{\alpha \in R_+} \frac{2 c_{\alpha}}{(\alpha, x)} \partial_{\alpha} + \sum\limits_{\alpha \in R_+}\frac{c_{\alpha} (\alpha, \alpha)}{(\alpha , x)^2} (1 - s_{\alpha}),
\end{equation}
one obtains the generalised Calogero--Moser Hamiltonian  
\begin{equation*}
 \Res_\pi \left( \sum\limits_{i = 1}^N \nabla_i^2 \right) = \Delta_y - \sum\limits_{\substack{\alpha \in R_+ \\ \widehat{\alpha} \neq 0}}\frac{2 c_{\alpha}}{(\widehat{\alpha} , y)} \partial_{\widehat{\alpha}},
\end{equation*}
where $y = (y_1, \ldots, y_n)$ are orthonormal coordinates on the space~$\pi$, $\Delta_y = \partial_{y_1}^2 + \cdots + \partial_{y_n}^2$, and $\widehat{\alpha}$ is the orthogonal projection of
$\alpha$ onto~$\pi$. The operators $\Res_\pi p(\nabla)$ for $p \in \C[x]^W$ give quantum integrals for this Hamiltonian.

\subsection{Matrix case}\label{spin case}
We can act with Dunkl operators also on the space $\widetilde{M} = U \otimes M$,
where $U$ is any complex vector space (and $\otimes = \otimes_\C$) called the vector part in what follows, where we make any element $h$ of the rational Cherednik algebra $\cH_c$ act as $\mathrm{id} \otimes h$ on $\widetilde{M}$. Assume now that $U$ is a right $W$-module and denote the action of $w \in W$ on $v \in U$ by $v \cdot w$. 

Let us denote by $\widetilde{M}^W$ the space of invariants under the diagonal (left) action of $W$ on $\widetilde{M}$ given by 
\begin{equation}\label{diag action}
    w \cdot (v \otimes f) = (v \cdot w^{-1}) \otimes (w f)
\end{equation}
for $w \in W$, $v \in U$, $f \in M$, and extended linearly. 
The space $\widetilde{M}^W$ is isomorphic to~$U \otimes_W M$, that is, the quotient of $\widetilde{M}$ by the relations
\begin{equation*}
    v \otimes (w f) =  (v \cdot w) \otimes f. 
\end{equation*}
Under the action of $\cH_c$ on $\widetilde{M}$ from above, the subalgebra $\cH_c^W$ preserves the space~$\widetilde{M}^W$.

We show below that $\cH_c^W$ can act not just on $\widetilde{M}^W$, but also on certain analytic vector-valued germs of functions defined on $D_{\Gamma_0}$ near generic points of $\pi$. 
Let $x_0 \in \pi$ be a generic point (meaning that if a Coxeter mirror of $W$ contains $x_0$ then it also contains the entire space~$\pi$), and let $Wx_0 = \cup_{w \in W} w(x_0)$ be the $W$-orbit of $x_0$. We define the space 
\[\mathcal{C}_{Wx_0}(U) = \mathcal{C}_{Wx_0}(D_{\Gamma_0}, U) = \bigoplus\limits_{x \in Wx_0} \mathcal{C}_{x}(D_{\Gamma_0}, U),\]
where $ \mathcal{C}_{x}(D_{\Gamma_0}, U)$ is
the space of germs of analytic functions defined on~$D_{\Gamma_0}$ near the point $x \in Wx_0$ 
with values in $U$. Note $\mathcal{C}_{x}(D_{\Gamma_0}, U) \cong$ $U \otimes \mathcal{C}_{x}(D_{\Gamma_0}, \C)$.
We have the following lemma (cf.\ \cite{Feigin} in the case when~$U$ is the trivial representation).

\begin{lemma} \label{invar} 
    Assume that the stratum $D_{\Gamma_{0}}$ defines an invariant parabolic ideal for the rational Cherednik algebra $\cH_c$.  Let $x_0 \in \pi$ be generic. Then the space $\mathcal{C}_{W x_0}(U)$ is an $\cH_c$-module.
\end{lemma}
\begin{proof}
     The elements of the Coxeter group $w \in W$ move the germs at one point to another one, namely, $w\colon  \mathcal{C}_{x}(D_{\Gamma_0}, U) \to \mathcal{C}_{w(x)}(D_{\Gamma_0}, U)$ for any $x \in W x_0$, with the action given by $(wF)(y) = F(w^{-1}y)$ for $y \in D_{\Gamma_0}$ near $w(x)$ and $F \in \mathcal{C}_{x}(D_{\Gamma_0}, U)$. The polynomial part of the rational 
    Cherednik algebra $\cH_c$ acts by multiplication. 
    The Dunkl operators act  on any analytic extension of the given direct sum of germs to a small $W$-invariant open set in the ambient space $V$ by formula~\eqref{Dunkl} with each reflection $s_\alpha$ acting by $\mathrm{id} \otimes s_\alpha$, and then we restrict the result to~$D_{\Gamma_0}$. Similarly to~\cite{Feigin}, the result does not depend on the choice of the extension.
\end{proof}

Let $\mathcal{C}_{Wx_0}^W(U)$ be the subset of those elements of $\mathcal{C}_{Wx_0}(U)$ that are fixed by the diagonal action of $W$ determined by $(wF)(y) = F(w^{-1}y) \cdot w^{-1}$ (similarly to the formula~\eqref{diag action}) for $F \in \mathcal{C}_{x}(D_{\Gamma_0}, U)$.
Then any element of $\mathcal{C}_{Wx_0}^W(U)$ is uniquely determined by the germ near the point $x_0$. In other words, $\mathcal{C}_{Wx_0}^W(U) \cong \mathcal{C}_{x_0}^{W_0}(D_{\Gamma_0},U) \cong  U^{W_0} \otimes \mathcal{C}_{x_0}(D_{\Gamma_0}, \C)$ as vector spaces, where $U^{W_0}$ is the subspace of vectors in $U$ fixed under the action of $W_0 \subseteq W$. Since $\mathcal{C}_{Wx_0}^W(U)$ is an $\cH_c^W$-module as a consequence of Lemma~\ref{invar}, we can also treat $\mathcal{C}_{x_0}^{W_0}(D_{\Gamma_0},U)$ as an $\cH_c^W$-module. We denote the action of an element $a \in \cH_c^W$ on $\mathcal{C}_{x_0}^{W_0}(D_{\Gamma_0},U)$ by $\widetilde{\Res}_{\pi} a$.

\begin{theorem}\label{maintheorem}
Assume that the stratum $D_{\Gamma_{0}}$ defines an invariant parabolic ideal for the rational Cherednik algebra $\cH_c$. Then the operator~$\sum_{i = 1}^N \nabla_i^2$ restricted to $\mathcal{C}_{x_0}^{W_0}(D_{\Gamma_0}, U)$
has the generalised spin Calogero--Moser form
\begin{equation}
\label{H2}
\begin{aligned}
    H_2 &= \widetilde{\Res}_{\pi}\left( \sum\limits_{i = 1}^N \nabla_i^2 \right) \\ &=
    \Delta_y - \sum\limits_{\substack{\alpha \in R_+ \\ \widehat{\alpha} \neq 0}}\frac{2 c_{\alpha}}{(\widehat{\alpha} , y)} \partial_{\widehat{\alpha}} + 
    \sum\limits_{\substack{\alpha \in R_{+} \\ \widehat{\alpha} \neq 0}}\frac{c_{\alpha} (\alpha, \alpha)}{(\widehat{\alpha}, y)^2}(1 - P_{\alpha}),
\end{aligned}
\end{equation}
where 
$P_{\alpha}$ denotes the action of the reflection $s_{\alpha} \in W$ on the vector space~$U$. Moreover, for any $p(\nabla) \in \mathbb{C}[\nabla_1, \ldots, \nabla_N]^W$, the operators
$
    \widetilde{\rm Res}_{\pi} p(\nabla)
$
pairwise commute for different choices of invariant polynomials, and in particular, all of them commute with the operator~\eqref{H2}.
\end{theorem}

\begin{proof}
The result follows immediately from equality~\eqref{q1} by similar arguments as in the proof of \cite[Theorem 5]{Feigin}.    
\end{proof}
Let us denote the higher integrals by
\begin{equation}
\label{higherham}
    H_p = \widetilde{\rm Res}_{\pi} p(\nabla) \qquad
    (p \in \mathbb{C}[x]^{W}).
\end{equation}
Now we would like to rewrite the operator \eqref{H2} in the potential gauge. 
Let $\widehat{R}_+ = \{ \widehat{\alpha} \mid \alpha \in R_+\}$.

\begin{theorem} \label{potential gauge}
Let $\pi \subset V$ be an intersection of mirrors
\begin{equation*}
    \pi = \{x \in V \mid (\beta, x) = 0, \ \forall \beta \in \Gamma^v_0 \}
\end{equation*}
corresponding to the Coxeter subgraph $\Gamma_0 \subset \Gamma$. 
Define the generalised coupling constants
\begin{equation*}
    \widehat{c}_{\widehat{\alpha}} = \sum\limits_{\substack{\gamma \in R_{+} \\ \widehat{\gamma} = \widehat{\alpha}}} c_{\gamma}.
\end{equation*}
Then
\begin{equation}
    \label{Ham}
    \begin{aligned}
        f^{-1}&\bigg( \Delta_y - \sum\limits_{\substack{\alpha \in R_+ \\ \widehat{\alpha} \neq 0}}\frac{2 c_{\alpha}}{(\widehat{\alpha} , y)} \partial_{\widehat{\alpha}} + 
        \sum\limits_{\substack{\alpha \in R_{+} \\ \widehat{\alpha} \neq 0}}\frac{c_{\alpha} (\alpha, \alpha)}{(\widehat{\alpha}, y)^2}(1 - P_{\alpha}) \bigg) f 
        \\
        &= \Delta_y - \sum\limits_{\widehat{\alpha} \in \widehat{R}_+ \setminus \{0\}} \frac{(\widehat{\alpha},\widehat{\alpha})}{(\widehat{\alpha}, y)^2} \bigg(\widehat{c}_{\widehat{\alpha}}\sum\limits_{\substack{\widehat\beta \in \widehat{R}_+ \setminus \{0\} \\\widehat{\beta} \sim \widehat{\alpha}}}\widehat{c}_{\widehat\beta} + \widehat{c}_{\widehat{\alpha}} \widehat{P}_{\widehat{\alpha}} \bigg) \eqqcolon L,
    \end{aligned}
\end{equation}
where $\sim$ denotes proportionality of vectors, $y \in \pi$, and
\begin{align*}
   & f = \prod\limits_{\widehat\alpha \in \widehat{R}_{+} \setminus \{0\} } (\widehat{\alpha}, y )^{\widehat{c}_{\widehat\alpha}}, \nonumber
    \\ 
    \widehat{c}_{\widehat{\alpha}} \widehat{P}_{\widehat{\alpha}} = \widehat{c}_{\widehat{\alpha}} +
    &\frac{1}{(\widehat{\alpha}, \widehat{\alpha})} \sum_{\substack{ \gamma \in R_+ \\\widehat{\gamma} = \widehat{\alpha}}} 
    c_{\gamma}(\gamma, \gamma) (P_{\gamma} - 1). 
\end{align*}
\end{theorem}

\begin{proof}
Similarly to~\cite[Proposition 2]{Feigin}, we compute that the left-hand side of~\eqref{Ham} equals
\begin{align*}
    \Delta_y &- \sum\limits_{\substack{\alpha \in R_+ \\ \widehat{\alpha} \neq 0}} \frac{c_{\alpha}(\widehat{\alpha},\widehat{\alpha})}{(\widehat{\alpha},y)^2} - \sum\limits_{\substack{\alpha \in R_+ \\ \widehat{\alpha} \neq 0}} \sum\limits_{\substack{\beta \in R_+ \\ 0 \neq \widehat{\beta} \sim \widehat{\alpha}}} \frac{c_{\alpha} c_{\beta}(\widehat{\alpha},\widehat{\alpha})}{(\widehat{\alpha},y)^2}  
    + \sum\limits_{\substack{\alpha \in R_+ \\ \widehat{\alpha} \neq 0}}  \frac{c_{\alpha} (\alpha, \alpha)}{(\widehat{\alpha},y)^2} (1 - P_{\alpha}).
\end{align*} 
Equality~\eqref{Ham} then follows by using
the definition of
$\widehat{P}_{\widehat{\alpha}}$ and $\widehat{c}_{\widehat{\alpha}}$.
\end{proof}

\begin{remark}
    Let us note that the elements
    \begin{equation*}
        S_{\widehat{\alpha}} = \sum_{\substack{ \gamma \in R_+ \\\widehat{\gamma} = \widehat{\alpha}}} 
        c_{\gamma}(\gamma, \gamma) s_{\gamma} \in \C[W]
    \end{equation*}
    lie in the centraliser of $\mathbb{C}[W_0]$ inside the group algebra of the Coxeter group $\mathbb{C}[W]$. Indeed, let $\beta \in \Gamma_0^v$, then
    \begin{equation*}
    \begin{aligned}
        [s_{\beta}, S_{\widehat{\alpha}}] &= 
        \sum\limits_{\substack{\gamma \in R_+ \\ \widehat{\gamma} = \widehat{\alpha}}}
        c_{\gamma} (\gamma, \gamma) [s_{\beta}, s_{\gamma}] \\
        &= 
        \sum\limits_{\substack{\gamma \in R_+ \\ \widehat{\gamma} = \widehat{\alpha}}}
        c_{\gamma} (\gamma, \gamma) (s_{s_{\beta}(\gamma)} - s_{\gamma}) s_{\beta} = 0,
    \end{aligned}
    \end{equation*}
    where the last equality is obtained by changing the index of summation in the first term to
    $\widetilde{\gamma} = s_{\beta} (\gamma)$, which is possible due to the fact that 
    the multiplicity function $c_{\gamma}$ is $W$-invariant and $\beta \in \Gamma_0^v$.
\end{remark}

In the particular case where $U = V$ is the reflection representation of the Coxeter group, $\widehat{P}_{\widehat{\alpha}}$ is 
a reflection in the space $V$ as we prove in the next proposition.
\begin{prop} \label{projref}
Let $U = V$ be the reflection representation of the Coxeter group $W$. 
Suppose that $\widehat{c}_{\widehat{\alpha}} \neq 0$ for  $\widehat{\alpha} \in \widehat{R}_+ \setminus \{0\}$. Then the operator
\begin{equation}
    \widehat{P}_{\widehat{\alpha}} = 1 + \frac{1}{\widehat{c}_{\widehat{\alpha}} (\widehat{\alpha}, \widehat{\alpha})} \sum_{\substack{ \gamma \in R_+ \\\widehat{\gamma} = \widehat{\alpha}}} 
    c_{\gamma}(\gamma, \gamma) (P_{\gamma} - 1)  \label{P hat}
\end{equation}
preserves
the space $\pi$, and when acting on $\pi$, it is equal to the orthogonal reflection 
about the hyperplane orthogonal to $\widehat{\alpha}$.
\end{prop}

\begin{proof}
Let us define 
\begin{equation*}
    \widetilde{P}_{\widehat{\alpha}} = \sum_{\substack{ \gamma \in R_+ \\ \widehat{\gamma} = \widehat{\alpha}}} 
     c_{\gamma}(\gamma, \gamma) P_{\gamma}.
\end{equation*}
For any $y \in \pi$, we have
\begin{equation}
\label{refl}
    \widetilde{P}_{\widehat{\alpha}}(y) 
    = \sum\limits_{\substack{\gamma \in R_{+} \\ \widehat{\gamma} = \widehat{\alpha}}}c_{\gamma}(\gamma, \gamma) y - 2\sum\limits_{\substack{\gamma \in R_{+} \\\widehat{\gamma} = \widehat{\alpha}}} c_{\gamma} (\gamma,y) \gamma. 
\end{equation}
To prove that the vector \eqref{refl} belongs to the subspace $\pi$, we need to simplify the second sum in \eqref{refl}. We claim that 
\begin{equation}
\label{elem}
     \sum\limits_{\substack{\gamma \in R_{+} \\ \widehat{\gamma} = \widehat{\alpha}}} c_{\gamma} (\gamma,y) \gamma =  \widehat{c}_{\widehat{\alpha}} (\widehat{\alpha},y) \widehat{\alpha}.
\end{equation}
Indeed, 
the set $S_{y} = \{c_{\gamma} (\gamma,y)\gamma \mid \gamma \in R_+, \, \widehat{\gamma} = \widehat{\alpha}\}$  
is $W_0$-invariant, where we use that $W_0$ is generated by reflections about the hyperplanes orthogonal to simple roots~$\Gamma_0^v$. 
Thus, the sum in the left-hand side of \eqref{elem} is fixed by $W_0$, hence belongs to the subspace $\pi$, and relation \eqref{elem} follows.

Now,
by using formula~\eqref{P hat} and equalities~\eqref{refl} and \eqref{elem},
we get that the action of the operator $\widehat{P}_{\widehat{\alpha}}$ on $\pi$ is
\begin{equation*}
    \widehat{P}_{\widehat{\alpha}} (y) = y +
    \frac{1}{\widehat{c}_{\widehat{\alpha}} (\widehat{\alpha}, \widehat{\alpha})} \bigg(  \widetilde{P}_{\widehat{\alpha}}(y)-\sum_{\substack{ \gamma \in R_+ \\\widehat{\gamma} = \widehat{\alpha}}} 
    c_{\gamma}(\gamma, \gamma) y \bigg)
    = y - \frac{2 (\widehat{\alpha}, y)}{(\widehat{\alpha}, \widehat{\alpha})} \widehat{\alpha},
\end{equation*}
which is the formula for the reflection on $\pi$ about the hyperplane orthogonal to~$\widehat{\alpha}$.
\end{proof}
\begin{remark} One can extend considerations of this subsection as follows. Let $\chi\colon W \to \mathbb C$ be a homomorphism. 
Instead of considering the space of invariants $\widetilde{M}^W$ one can take the subspace $\widetilde{M}^{W, \chi} \subset \widetilde{M}$ consisting of elements $a \in \widetilde{M}$ such that
\begin{equation*}
    w(a) = \chi(w) a.
\end{equation*}
Equivalently, $\widetilde{M}^{W, \chi}$ can be defined as the quotient of $\widetilde{M}$ by the sum of images ${\rm Im}(w - \chi(w))$ over $w \in W$. Then, analogues of Theorems~\ref{maintheorem} and \ref{potential gauge} take place with a replacement of $P_{\alpha}$ by $\chi(s_{\alpha}) P_{\alpha}$.

\end{remark}
Particular cases of the operator \eqref{Ham} are related to the generalised matrix Calogero--Moser Hamiltonians found in~\cite{CGV}, as we explain in the next subsection. 

\subsection{Restricted operators from classical root systems}
In this subsection, we write down the restricted Hamiltonians~\eqref{Ham} explicitly in the case where $R$ is a classical root system. 

\subsubsection{$A$-series} \label{rat type A}
 
We first briefly review how the scalar version of the so-called deformed Calogero--Moser system can be obtained using parabolic submodules of the rational Cherednik algebra of type $GL_N$, that is, for the symmetric group~$S_N$ acting on $V = \C^N$ \cite{Feigin}. Note that Theorem~\ref{invariance} as well as Lemma~\ref{invar} remain true when the rank of the root system is less than the dimension of the ambient space $V$. 

Let $c = \frac{1}{k}$ for an integer $k \geq 2$. Recall that, in this case, the polynomial representation $\mathbb{C}[x]$ becomes reducible. Namely, let us define 
\begin{equation*}
    \pi_{m,k} = \{(x_1, \dots, x_N) \in V \mid x_{1}= \dots = x_{k},  \ldots ,  x_{(m-1)k + 1}= \dots = x_{mk}  \},
\end{equation*}
$m \in \Z_{\geq 1}$, $mk \leq N$, 
and 
\begin{equation}
\label{strata}
    D_{m, k} = \underset{w \in S_N}{\bigcup} w(\pi_{m, k}),
\end{equation}
and denote by $I_{m ,k}$ the ideal of the polynomials vanishing on $D_{m,k}$.
Then the ideal $I_{m, k}$ is a submodule inside the polynomial module of the rational Cherednik algebra. Hence, the quotient $M = \mathbb{C}[x]/I_{m, k}$ is also its module. This corresponds to taking 
\begin{equation*}
    \Gamma_0^v = \{ e_i - e_{i+1} \mid i = jk+1, \dots, (j+1)k -1, \ 0 \leq j \leq m-1 \}
\end{equation*}
in Section~\ref{invariant ideals}.

Symmetric polynomials of Dunkl 
operators preserve the space of invariants $M^{S_N}$. The following operators lead to the integrable Hamiltonian and its integrals for the so-called deformed Calogero--Moser system:
\begin{equation*}
    \Res \left(\sum\limits_{i = 1}^N \nabla_{i}^j\right) \quad (j \in \mathbb{Z}_{\geq 0}), 
\end{equation*}
where $\Res$ denotes the restriction to the space $M^{S_N}$ \cite{Feigin}.

In the spin case, we consider the module $\widetilde{M} = U \otimes (\mathbb{C}[x]/I_{m, k})$, where $U$ is a vector space on which the rational Cherednik algebra acts trivially, and we assume that $U$ is a right $S_N$-module.  
We define $\widetilde{M}^{S_N}$ as the space of invariants under the (diagonal) action of $W = S_N$ on $\widetilde M$ given by~\eqref{diag action}.
This is a subspace such that the action of the symmetric group on the polynomial part is equivalent to acting 
on the vector part. Symmetric polynomials of Dunkl operators preserve the space
$\widetilde{M}^{S_N}$. 
Let $\widetilde{\Res}_\pi$ be defined as in Section~\ref{spin case} with $\pi = \pi_{m,k}$,
and let us define the 
following commuting operators
\begin{equation*}
    H_j = \widetilde{\Res}_\pi \left(\sum\limits_{i = 1}^N \nabla_i^j\right).
\end{equation*}

For every group of coordinates $x_{jk + 1} = x_{jk + 2} = \dots = x_{jk + k}$ ($j = 0 , \ldots, m-1$), let us define new coordinates
\begin{align*}
    &z_{jk + 1} = \frac{x_{jk + 1} + \dots + x_{jk + k}}{\sqrt{k}}, \quad  z_{jk + 2}  = \frac{x_{jk + 1} - x_{jk + 2}}{\sqrt{k}}, 
    \\
     &z_{jk + 3} = \frac{x_{jk + 2} - x_{jk + 3}}{\sqrt{k}}, \ \ldots \ , \quad  z_{jk + k} = \frac{x_{jk + k -1} - x_{jk + k}}{\sqrt{k}},
\end{align*}
while the other $x$-coordinates remain unchanged. On $\pi_{m,k}$, most of the new coordinates vanish. 
In what follows, we will rename the non-zero $z$-coordinates to $y_1, \dots, y_m$, namely $y_{j+1} = z_{jk + 1}$ for $j = 0, \dots, m-1$, and the remaining $x$-coordinates will be relabelled to $y_{m+j} = x_{mk+j}$ for $j=1, \dots, n$, where $n = N - mk$.
Then the first integral is simply a momentum operator
\begin{equation*}
    H_1 = \sqrt{k} \sum\limits_{i = 1}^m  \partial_{y_i} + \sum\limits_{i = m+ 1}^{m + n} \partial_{y_i}.
\end{equation*}

\begin{prop}\label{rat A prop}
The Hamiltonian $H_2$ takes the form of the deformed spin Calogero--Moser operator
\begin{align}
    &\sum\limits_{i = 1}^{m + n} \partial_{y_i}^2 + 
    \sum\limits_{m + 1 \leq i < j \leq m+n} \frac{2(1 - \widetilde{P}_{ij})}{k(y_i - y_j)^2} + 
    \sum\limits_{i = 1}^m \sum\limits_{j = m + 1}^{m + n}
     \frac{2(k - \widetilde{P}_{ij})}{(y_i - \sqrt{k}y_j)^2}  \label{Hamiltonian}
   \\ 
    &\ + \sum\limits_{1 \leq i < j \leq m} \frac{2(k^2 - \widetilde{P}_{i j})}{(y_{i} - y_{j})^2} - 
    \sum\limits_{m+1 \leq i < j \leq m+n}\frac{2}{k(y_i - y_j)} (\partial_{y_i} - \partial_{y_j}) \nonumber 
    \\ 
    &\ - \sum\limits_{i = 1}^m \sum\limits_{j = m + 1}^{m + n} \frac{2}{y_i - \sqrt{k} y_j} 
    (\partial_{y_{i}} - \sqrt{k}\partial_{y_j}) 
    - \sum\limits_{1 \leq i < j \leq m} \frac{2k}{y_{i} - y_{j}} (\partial_{y_{i}} - 
    \partial_{y_{j}}), \nonumber
\end{align}
with
\begin{equation}
\label{perm}
    \widetilde{P}_{ij} = 
    \begin{cases}
        P_{m(k-1)+i, \, m(k-1)+j}, & m + 1 \leq i < j \leq m + n, \\
        \sum\limits_{r = (i - 1)k +1}^{i k}P_{r, \, m(k-1)+j}, & 1 \leq i \leq m, \ m + 1 \leq j \leq m + n, \\
        \sum\limits_{r = (i - 1)k +1}^{i k}\sum\limits_{s = (j - 1)k +1}^{j k} P_{r s}, & 1 \leq i < j \leq m,
    \end{cases}
\end{equation}
where $P_{ij}$ denotes the action of the transposition $(i,j) \in S_N$ on the vector space $U$.
\end{prop}
\begin{proof}
We use formula \eqref{H2} with the above-introduced $y$-coordinates on $\pi = \pi_{m,k}$. 
Consider the non-zero vectors in the projection onto~$\pi$ of the positive half of the $A$-type root system $A_{N-1}^+ = \{e_i - e_j \mid 1 \leq i < j \leq N \}$. Then for any one of them, we have $1$, $k$, or $k^2$ roots in $A_{N-1}^+$ that project to it. By combining all the terms for the roots in $A_{N-1}^+$ that project to the same vector, we get the result~\eqref{Hamiltonian}. 
\end{proof}
In the potential gauge, the Hamiltonian \eqref{Hamiltonian} takes the form 
\begin{equation}\label{Hamiltonian alt gauge}
    \begin{aligned}
        L &= \sum\limits_{i = 1}^{m + n} \partial_{y_i}^2 - \sum\limits_{m + 1 \leq i < j \leq m + n} \frac{2 (k^{-2} + k^{-1} \widetilde{P}_{ij})}{(y_i - y_j)^2} 
        \\
        &\quad - \sum\limits_{i = 1}^m \sum\limits_{j = m + 1}^{m + n}\frac{2(1 + \widetilde{P}_{ij})}{(y_i - \sqrt{k}y_j)^2} - \sum\limits_{1 \leq i < j \leq m} \frac{2(k + \widetilde{P}_{ij})}{(y_i - y_j)^2}
    \end{aligned}
\end{equation}
by using Theorem~\ref{potential gauge}.

We now make a connection with the deformed matrix Calogero--Moser operators introduced by Chalykh, Goncharenko, 
and Veselov~\cite{CGV}. The Hamiltonian \eqref{Hamiltonian alt gauge} in the case where $U=V$ is the reflection representation and $N= mk + 1$ (i.e., $n=1$) is related to a Hamiltonian
introduced in \cite{CGV}, as we now explain. 

Let $A_{m,n;k}^+$ be a positive part of the deformed root system of type~$A$ in the sense of Sergeev and Veselov \cite{SV} (see also~\cite{Serganova}):
\begin{align*}
    A_{m,n; k}^+ &= \{ \tfrac{1}{\sqrt{k}}(\widetilde{e}_i - \widetilde{e}_j) \mid 1 \leq i < j \leq m \}
    \\
    &\qquad \cup \{ \tfrac{1}{\sqrt{k}} \widetilde{e}_i - \widetilde{e}_j  \mid 1 \leq i \leq m, \; m+1 \leq j \leq m+n\}
    \\
    &\qquad \cup  \{ \widetilde{e}_i - \widetilde{e}_j \mid m+1 \leq i < j \leq m+n \}  \subset V_{m,n} \cong 
    \mathbb{C}^{m + n},
\end{align*}
where $\widetilde{e}_1, \ldots, \widetilde{e}_{m+n}$ is an orthonormal basis for $V_{m,n}$.

Assuming that $k \in \mathbb{Z}_{\geq 2}$ and $N = mk + n$,  
let us define the linear isometry $\phi \colon V_{m,n} \to V$ by
\begin{equation*}
    \begin{array}{cc}
        \phi(\widetilde{e}_{i}) = \frac{1}{\sqrt{k}}\sum\limits_{j \in J_{i}} e_j, \quad 1 \leq i \leq m,
        & \quad  \phi(\widetilde{e}_{m+i}) = e_{mk+i}, \quad 1 \leq i \leq n,
    \end{array}
\end{equation*}
where $J_i = \{(i-1)k + 1, \, (i-1)k + 2, \dots, ik \}$.
Then the image $\phi(A_{m,n; k}^+)$ is the set of non-zero vectors in the projection onto~$\pi$ of $A_{N-1}^+$.
In the above-introduced $y$-coordinates on $\pi$, we have $(\phi(\widetilde{e}_i), x) = y_i$ for all $i = 1, \dots, m+n$.

Let $s_{ij}$ denote the reflection on $V$ about the hyperplane orthogonal to the vector $e_i - e_j \in A_{N-1}^+$, and let $\widetilde{s}_{ij}$ denote the reflection on $V_{m,n}$ about the hyperplane orthogonal to the vector
in $A_{m,n; k}^+$ corresponding to the indices $i$, $j$. 
Consider the map $g \colon \{ \widetilde{s}_{ij} \} \to \C[S_N]$ given by
\begin{align*}
    g(\widetilde{s}_{m+i, \, m+j}) &= s_{mk+i, \, mk+j}, \quad 1 \leq i < j \leq n, 
    \\
    g(\widetilde{s}_{i, \, m+j}) &= \frac{1-k}{1+k} + \frac{2}{1+k} \sum\limits_{l \in J_{i}} s_{l, \, mk+j}, \quad 
    1 \leq i \leq m, \ 1 \leq j \leq n,
    \\
    g(\widetilde{s}_{ij}) &= 1-k+\frac{1}{k} \sum\limits_{r \in J_{i}} \sum\limits_{l \in J_{j}} s_{rl}, \quad 1 \leq i < j \leq m.
\end{align*}
The map $g$ sends the reflection about the hyperplane orthogonal to $v \in A_{m,n; k}^+$ to the operator  $\widehat{P}_{\widehat{\alpha}}$ given by~\eqref{P hat} with $\widehat{\alpha} = \phi(v)$.  Proposition~\ref{projref} claims that when acting on the subspace 
\begin{equation*}
    \phi(V_{m,n}) = \Span\{e_1 + \dots + e_k, \ldots, e_{(m-1)k+ 1} + \dots + e_{mk}, e_{mk+1}, \ldots, e_N\} = \pi,
\end{equation*}
the operator $\widehat{P}_{\widehat{\alpha}}$ coincides with the reflection about the hyperplane orthogonal to $\widehat{\alpha}$.  This implies the following statement.

\begin{prop}
\label{prop37}
The maps $\phi$ and $g$ defined above, together with the actions of the group $S_N$ and the reflections
$\widetilde{s}_{ij}$ $(1 \leq i < j \leq m+n)$, make the following diagram commute
\begin{equation*}
\begin{tikzcd}[row sep=large,column sep=huge]
    V_{m,n}
        \arrow[r, "\phi"]
        \arrow[d, "\widetilde{s}_{ij}"] &
    \pi \subset V
        \arrow[d, "g(\widetilde{s}_{ij})"] \\
    V_{m,n} 
        \arrow[r, "\phi"] &
   \pi \subset V
\end{tikzcd}
\end{equation*}
\end{prop}

Using Propositions \ref{projref}, \ref{prop37} and Theorem~\ref{potential gauge}, one can see that when $U=V$ and $n=1$, the potential gauge form of the Hamiltonian~\eqref{Hamiltonian} is related to a matrix operator defined in \cite{CGV} as follows.

Given a configuration $\mathcal A$ with a multiplicity function $c\colon {\mathcal A}\to {\mathbb Z}$ the operator 
$L_{\mathcal A}^\varepsilon$ of the form 
\begin{equation}
\label{Lepsilon}
L_{\mathcal A}^\varepsilon=\Delta_y - \sum_{\alpha\in \mathcal A} \frac{c_\alpha(c_\alpha - \varepsilon P_\alpha)(\alpha, \alpha)}{(\alpha,y)} 
\end{equation}
with $\varepsilon =1$ 
was considered in \cite{CGV}. $D$-integrability of such operators was established in \cite{CGV} for root systems as well as for ${\mathcal A}=A^+_{m,1; k}$ with multiplicities described above. In a similar way one can establish $D$-integrability of the corresponding operator $L_{\mathcal A}^\varepsilon$ with $\varepsilon=-1$ which coincides with the operator~\eqref{Hamiltonian alt gauge}. It is unclear to us whether one can realise such an operator with $\varepsilon =1$ via a restriction-type procedure similar to the one described in this paper.

\subsubsection{$B$-series}\label{rat type B}
A positive half of the root system $B_N$ is $B_N^+ = \{ e_i \pm e_j \mid 1 \leq i < j \leq N\} \cup \{ e_i  \mid i= 1, \dots, N\}$.
Let $c_1$ be the multiplicity of the roots $e_i \pm e_j$ and $c_2$ the multiplicity of the roots $e_i$. Then the most general parabolic 
stratum leading to an invariant parabolic ideal can be defined as follows. 

Let $\pi_{m,k;\, l} \subset \mathbb{C}^N$ for $mk+l \leq N$ be defined by the equations
\begin{align*}
    &x_1 = x_2 = \dots = x_k, \ x_{k+1} = x_{k + 2} = \dots = x_{2k}, \ \dots, \\
    &x_{(m-1)k+1} = x_{(m-1)k+2} = \dots = x_{mk}, \\
    &x_{N - l + 1} = x_{N - l + 2} = \dots = x_{N} = 0,    
\end{align*}
where we allow $k=1$ or $l=0$, which means that the corresponding relations are absent. The corresponding parabolic stratum is then defined by
\begin{equation*}
    D_{m,k; \, l} = \bigcup\limits_{w \in \mathcal{B}_N} w(\pi_{m,k; \, l}),
\end{equation*}
where $\mathcal{B}_N$ is the Weyl group of type $B_N$.
For the corresponding parabolic ideal to be invariant,
the multiplicities must satisfy the following. When $k > 1$ and $l \neq 0$, we must have $c_1 = \frac{1}{k}$ and $c_2 = \frac{1}{2} - \frac{l - 1}{k}$. When $k > 1$ and $l = 0$, we must have $c_1 = \frac{1}{k}$ and there is no restriction on $c_2$. When $k = 1$ and $l \neq 0$, we must have $2c_1(l-1) + 2c_2 = 1$.

The non-zero vectors in the projection of $B_N^+$ onto the space~$\pi_{m,k; \, l}$ are 
\begin{align}
        B_{m,n; \, k}^+ &= \{\widetilde{e}_i \pm \widetilde{e}_j \mid m+1 \leq i < j \leq m+n \} \label{projroot} \\
        &\qquad \cup \{\tfrac{\widetilde{e}_i}{\sqrt{k}} \pm \widetilde{e}_j \mid 1 \leq i \leq m, \ m+1 \leq j \leq m+n\} \nonumber
        \\ 
        &\qquad \cup\{\tfrac{\widetilde{e}_i \pm \widetilde{e}_j}{\sqrt{k}} \mid 1 \leq i < j \leq m \} \cup \{\widetilde{e}_i \mid m+1 \leq i \leq m+n \} \nonumber
        \\
        &\qquad \cup\{\tfrac{\widetilde{e}_i}{\sqrt{k}} \mid 1 \leq i \leq m\} \cup \{\tfrac{2 \widetilde{e}_i}{\sqrt{k}} \mid 1 \leq i \leq m \}, \nonumber
\end{align}
where $n = N - mk - l$ and 
\begin{equation*}
    \begin{array}{cc}
        \widetilde{e}_{i} = \frac{1}{\sqrt{k}}\sum\limits_{j \in J_{i}} e_j, \quad 1 \leq i \leq m,
        & \quad  \widetilde{e}_{m+i} = e_{mk+i}, \quad 1 \leq i \leq n,
    \end{array}
\end{equation*}
with $J_i = \{(i-1)k + 1, \, (i-1)k + 2, \dots, ik \}$, and respective generalised multiplicities are 
\begin{alignat*}{3}
    &\widehat{c}_{\widetilde{e}_i \pm \widetilde{e}_j} = c_1, \qquad 
    &&\widehat{c}_{(\sqrt{k})^{-1}\widetilde{e}_i \pm \widetilde{e}_j} = c_1 k, \qquad 
    &&\widehat{c}_{(\sqrt{k})^{-1}(\widetilde{e}_i \pm \widetilde{e}_j)} = c_1 k^2, \\
    &\widehat{c}_{\widetilde{e}_i} = 2c_1 l + c_2, \qquad 
    &&\widehat{c}_{(\sqrt{k})^{-1}\widetilde{e}_i} = 2c_1 kl + c_2 k, \qquad
    &&\widehat{c}_{2(\sqrt{k})^{-1}\widetilde{e}_i} = \frac{k-1}{2}. 
\end{alignat*}
In the case $k=1$, the vectors $\{\tfrac{2 \widetilde{e}_i}{\sqrt{k}} \mid 1 \leq i \leq m \}$ should be excluded from $B_{m,n; \, k}^+$.


The following proposition is a corollary of Theorem~\ref{potential gauge}.
\begin{prop}

Let $y_1, \dots, y_m$ be as in Section~\ref{rat type A}, and let $y_{m+i} = x_{mk+i}$ for $i=1, \dots, n$.
The generalised spin Calogero--Moser Hamiltonian (in potential gauge) corresponding to the above-defined stratum  $\pi_{m,k;\, l}$ is
\begin{align}
    L &= \sum_{i=1}^{m+n} \partial_{y_i}^2 - \sum\limits_{m+1 \leq i < j \leq m+n} 2 c_1 \left(\frac{c_1 + \widehat{P}_{\widetilde{e}_i-\widetilde{e}_j}}{(y_i - y_j)^2} + \frac{c_1 + \widehat{P}_{\widetilde{e}_i+ \widetilde{e}_j}}{(y_i + y_j)^2} \right) \label{DeformedHam}
    \\
    & \ -\sum\limits_{1 \leq i < j \leq m} 2 c_1 k^2\left(\frac{c_1 k^2 + \widehat{P}_{\frac{\widetilde{e}_i - \widetilde{e}_j}{\sqrt{k}}}}{(y_i -y_j)^2} + \frac{c_1 k^2 + \widehat{P}_{\frac{\widetilde{e}_i + \widetilde{e}_j}{\sqrt{k}}}}{(y_i + y_j)^2} \right) \nonumber
    \\
    & \ - \sum\limits_{i = 1}^m \sum\limits_{j = m+1}^{m+n} c_1 k(k+1) \left(\frac{c_1 k +  \widehat{P}_{\frac{\widetilde{e}_i}{\sqrt{k}}-\widetilde{e}_j}}{(y_i - \sqrt{k} y_j)^2} + \frac{c_1 k + \widehat{P}_{\frac{\widetilde{e}_i}{\sqrt{k}} + \widetilde{e}_j}}{(y_i + \sqrt{k} y_j)^2}  \right) \nonumber
    \\ 
    & \ -\sum\limits_{ i = m+1}^{m+n} \frac{(2 c_1 l + c_2)^2 + (2 c_1 l + c_2)\widehat{P}_{\widetilde{e}_i}}{y_i^2} \nonumber
    \\
    & \ -\sum\limits_{i = 1}^m \frac{1}{y_i^2}\left(\left(2c_1kl + c_2k + \frac{k-1}{2}\right)^2 
     + 
     k \left(2c_1 l  +  c_2\right)\widehat{P}_{\frac{\widetilde{e}_i}{\sqrt{k}}}
      + \frac{k-1}{2}  \widehat{P}_{\frac{2\widetilde{e}_i}{\sqrt{k}}}\right), \nonumber
\end{align}
where $c_2$ is a free parameter if $l = 0$ and otherwise $c_2 = \frac{1}{2} - c_1(l - 1)$, and $c_1$ is a free parameter if $k=1$ and otherwise $c_1 = \frac1k$. 
\end{prop}

Note that by the previous proposition, the projection of $B_N^+$ onto $\pi_{m,k; \, l}$ leads for $k=1$ to a projected system of type~$B$, and for $k>1$, to a special case of Sergeev--Veselov's deformed $BC$-type root system~\cite{SV} (which includes the standard $BC$ root system).

The Hamiltonian \eqref{DeformedHam} for $U=V$ being the reflection representation, $n=1$, $l=0$ ($N= mk + 1$), $c_2 \in \Z_{\geq 0}$, and odd $k$ is related to a Hamiltonian
introduced in \cite{CGV} for a deformation $C_{m+1}(k c_2+\frac{k-1}2, c_2)$ of the root system~$C_{m+1}$. That is, such a Hamiltonian \eqref{DeformedHam} coincides with the operator \eqref{Lepsilon} with $\varepsilon=-1$ and configuration ${\mathcal A} = B_{m,1;k}^+ = C_{m+1}(k c_2+\frac{k-1}2, c_2)$, $l=0$, $c_1=\frac1k$.

\begin{example}\label{B_2 example}
    Let us consider the case $N=3$, $k=1$, $l=1$, and let $U$ be the two-dimensional irreducible representation of $\mathcal{B}_3$ with basis $x_1^2 -x_2^2$ and $x_2^2 - x_3^2$, with the natural action of $\mathcal{B}_3$ on polynomials.
    Then $U^{W_0} = U$ as $W_0 = \langle s_{e_3} \rangle$, and $\widehat{R}_+ \setminus \{0\} = B_2^+$. 
    The restricted Hamiltonian is
    \begin{equation*}
        L = \partial_{x_1}^2 + \partial_{x_2}^2 - \sum_{i=1}^2\frac{\widehat{c}_{e_i}(\widehat{c}_{e_i} + \widehat{P}_{e_i})}{x_i^2}
       -\sum_{\varepsilon \in \{\pm 1\}}\frac{2c_1(c_1 + \widehat{P}_{e_1+ \varepsilon e_2})}{(x_1+ \varepsilon x_2)^2},
    \end{equation*}
    where $\widehat{c}_{e_i} = 2c_1 + \frac{1}{2}$, and
    \begin{equation*}
        \begin{array}{c}
        \widehat{P}_{e_1} = \frac{1}{4c_1 + 1} \begin{pmatrix}
            1-4c_1 & -8c_1 \\
            -8c_1 & 1-4c_1
        \end{pmatrix},
        \qquad \widehat{P}_{e_2} = \frac{1}{4c_1 + 1} \begin{pmatrix}
            4c_1 + 1 & 0 \\
            8c_1 & 1-12c_1
        \end{pmatrix}, \\ \ \\
        \widehat{P}_{e_1 \pm e_2} = \begin{pmatrix}
            -1 & 1 \\
            0 & 1
        \end{pmatrix}.
        \end{array}
    \end{equation*}
    The Hamiltonian $L$ has a $4$th-order quantum integral that can be obtained as $\widetilde{\Res}_{\pi}( \sum_{i = 1}^3 \nabla_i^4)$, which we write out explicitly in Appendix~\ref{B_2 integral} for illustration for $c_1 = 1$. 
    
    We note that for generic $c_1$, there is no representation $\varphi$ of $\mathcal{B}_2$ on~$\C^2$ and multiplicities $m,n \in \C$ such that the above $L$ would coincide with the standard spin $B_2$ Calogero--Moser operator
    \begin{equation*}
        \partial_{x_1}^2 + \partial_{x_2}^2 - \sum_{i=1}^2\frac{m(m + \varphi(s_{e_i}))}{x_i^2}
       -\sum_{\varepsilon \in \{\pm 1\}}\frac{2n(n + \varphi(s_{e_1+ \varepsilon e_2}))}{(x_1+ \varepsilon x_2)^2}.
    \end{equation*}
\end{example}

\subsubsection{$D$-series}\label{D series}
A positive half of the root system $D_N$ is $D_N^+ = \{ e_i \pm e_j \mid 1 \leq i < j \leq N\}$ with the same multiplicity $c$ for all vectors.
In type~$D$, there exist three different types of parabolic strata leading to invariant parabolic ideals. The first type corresponds to the subspace $\pi_{m,k}^{\epsilon}$ ($\epsilon \in \{ \pm \}$, $mk \leq N$) defined by the equations
\begin{align*}
    &x_1 = x_2 = \dots = x_k, \ x_{k+1} = x_{k+2} = \dots = x_{2k}, \ \ldots, \\
    &x_{(m-1)k + 1} = x_{(m-1)k + 2} = \dots = \epsilon x_{mk}
\end{align*}
for $c = \frac{1}{k}$. 
The two cases $\epsilon = \pm$ lead to equivalent projected systems, so it suffices to consider the case $\epsilon = +$. 
The second type of stratum corresponds to the subspace $\pi_p$ given by
\begin{equation*}
     x_{N-p+1} = x_{N-p+2} = \dots = x_N = 0
\end{equation*}
for $c = \frac{1}{2(p-1)}$ and $p \in \Z_{\geq 2}$. The last type is the intersection of the first two types, and corresponds to the subspace $\pi_{m,k}$ defined by
\begin{align*}
    &x_1 = x_2 = \dots = x_k, \ x_{k+1} = x_{k+2} = \dots = x_{2k}, \ \ldots, \\
    &x_{(m-1)k+1} = x_{(m-1)k + 2} = \dots = x_{mk}, \\
    &x_{N-\frac{k}{2}} = x_{N-\frac{k}{2}+1} = \dots = x_N = 0
\end{align*}
for even $k$, $c = \frac{1}{k}$, and $mk + \frac{k}{2} + 1 \leq N$. 

The non-zero vectors in the respective projections of~$D_N^+$ (and their generalised multiplicities)
are particular cases of projections of $B_N^+$ (recall that the latter system's multiplicities were called $c_1$, $c_2$). Namely, for $\pi_{m,k}^{+}$, $\pi_p$, and~$\pi_{m,k}$, they are respectively
\begin{equation*}
    \begin{array}{c}
    \displaystyle{
    B_{m,n; \, k}^+, \quad {\rm with} \; c_1 = \frac{1}{k}, \; c_2 = 0, \; n = N - mk \ (\text{i.e., } l=0)
    }
    \\ \ \\
    \displaystyle{
    B^{+}_{N - p}, \quad {\rm with} \; c_1 = \frac{1}{2(p-1)}, \; c_2 = \frac{p}{p-1},
    }
    \\ \ \\
    \displaystyle{
    B_{m,n; \, k}^+, \quad {\rm with} \; c_1 = \frac{1}{k}, \; c_2 = 0, \; n = N - mk -\frac{k}{2}-1 \ (\text{i.e., } l=\frac{k}{2}+1),
    }
    \end{array}
\end{equation*}
where those vectors that have zero generalised multiplicity for the given $c_1$ and $c_2$ are thought of as excluded from $B_{m,n; \, k}^+$.

In all the cases, the corresponding Hamiltonians are of the form \eqref{DeformedHam}, with the respective generalised multiplicities. For each case, the computation of the operators $\widehat{P}_{\widehat{\alpha}}$ uses formula \eqref{P hat}. Let us note that in the case of $U$ being the reflection representation~$V$, the resulting operators depend only on the projected roots.

\subsection{Systems with a harmonic potential}
Let us show how the methods presented above allow us to obtain generalised
spin Calogero--Moser systems with a harmonic potential. In this part of the paper, we consider only classical root systems. Our considerations 
follow the logic of scalar case in~\cite{Feigin}. Let us introduce the operators
\begin{equation*}
    \nabla^{\pm}_i = \nabla_i \pm  \omega x_i, \quad h_i = \nabla_i^{+} \nabla_i^{-}.
\end{equation*}
For classical root systems $R$ and the associated Weyl group $W$, the elements 
\begin{equation} \label{harmonicintegrals}
    K_m = \sum\limits_{i = 1}^N h_i^m \qquad (m \in \Z_{\ge 0})
\end{equation}
are $W$-invariant and
pairwise commute (see \cite{Polych} for type $A$, and \cite{Feigin} and references therein for types $B, D$). To obtain the generalised spin Calogero--Moser Hamiltonian with a harmonic
potential, we compute
\begin{equation} \label{Hamiltonianharmonic}
    K_1 = \sum\limits_{i = 1}^N \left(\nabla_i^2 - \omega^2 x_i^2\right) + 
    2 \omega \sum\limits_{\alpha \in R_+} c_{\alpha} s_{\alpha} - \omega N,
\end{equation}
where we used the commutation relations
\begin{equation*}
    [\nabla_i, x_j] = \delta_{ij} - \sum\limits_{\alpha \in R_+} c_{\alpha}
    (\alpha, e_i) (\alpha^{\vee}, e_j)s_{\alpha},
\end{equation*}
$\alpha^\vee = 2\alpha/(\alpha, \alpha)$. Notice that the element $\sum_{\alpha \in R_+} c_{\alpha} s_{\alpha}$ belongs to the centre of $\mathbb{C}[W]$, and it follows that 
\begin{equation*}
    \widetilde{K}_1 = \sum\limits_{i = 1}^N \left(\nabla_i^2 - \omega^2 x_i^2\right)
\end{equation*}
is $W$-invariant and commutes with all $K_m$.
Using formula \eqref{Hamiltonianharmonic} and the operators \eqref{harmonicintegrals}, we act on $\widetilde{M}^W$ or analytic vector-valued germs, and we obtain 
the following result. 
\begin{theorem}
    Let $R$ be the root system $A_{N - 1}$, $B_N$ or $D_N$, and~$\widetilde{M}^W$ be the 
    module defined previously by the stratum $D_{\Gamma_0}$ and the right Weyl group 
    module $U$. Then the system defined by the generalised spin Calogero--Moser 
    Hamiltonian with a harmonic term 
    \begin{equation*}
        H^{\rm hm} =\widetilde{\mathrm{Res}}_{\pi}\widetilde{K}_1 = H_2 - \omega^2 (y , y) ,
    \end{equation*}
    where $H_2$ is the Hamiltonian of the generalised spin Calogero--Moser system~\eqref{H2}, admits pairwise commuting quantum integrals
    \begin{equation*}
        H_m^{\rm hm} = \widetilde{\mathrm{Res}}_{\pi} K_m,
    \end{equation*}
where $m\in \mathbb \Z_{\ge 0}$.
\end{theorem}

\section{Generalised trigonometric spin Calogero--Moser systems}\label{trig case}
In this section, we continue to use the notations introduced in Sections~\ref{invariant ideals} and~\ref{rational case}, but $R$ will now be a reduced crystallographic root system and~$W$ the corresponding Weyl group (the case of the non-reduced root system $BC_N$ will be considered in Section~\ref{BC}).

The trigonometric Cherednik algebra associated with~$R$ can be defined by its faithful action on the group algebra $\C[\{e^{(\alpha , x)} \mid \alpha \in P \}]$ of the weight lattice $P$. It is generated by~$W$,~$e^{(\alpha , x)}$ ($\alpha \in P$), and 
the commuting trigonometric Dunkl operators given by
\begin{equation}
    \label{trigDunkl}
    \nabla^{\rm tr}_{\xi} = \partial_{\xi} - \sum\limits_{\alpha \in R_+} \frac{c_{\alpha} (\alpha, \xi)}{1 - e^{-(\alpha,x)}}(1 - s_{\alpha}) + (\rho , \xi),
\end{equation}
where $\xi \in V$ and $\rho = \frac{1}{2}\sum_{\alpha \in R_+}c_{\alpha} \alpha$ is a generalised Weyl vector \cite{Ch1}. 
 
Define the subspace $\pi \subset V$ as in the rational case (Section~\ref{invariant ideals}). Let~$x_0$ be a generic point of $\pi$, meaning that if $e^{(\alpha, x_0)} = 1$ for some $\alpha \in R$ then $\alpha \in \Span \Gamma_0^v$.

The trigonometric Cherednik algebra has an action on the space $\mathcal{C}_{Wx_0}(V, U) = \oplus_{x \in Wx_0} \mathcal{C}_{x}(V, U)$ of germs on $V$, defined analogously to the rational Cherednik algebra case. The next theorem gives the conditions under which
this module has a submodule $\mathcal{I}_{\Gamma_0}$ consisting of those elements of $\mathcal{C}_{Wx_0}(V, U)$ that vanish when restricted to $D_{\Gamma_0}$.
In that case, we call~$\mathcal{I}_{\Gamma_0}$ a parabolic submodule, in analogy with the parabolic ideal~$I_{\Gamma_0}$.

\begin{theorem} 
    Let $\Gamma_0 = \coprod_{i = 1}^l \Gamma_i$ be the decomposition
    of the subgraph into connected components. Then the space $\mathcal{I}_{\Gamma_0}$ is invariant under the trigonometric Cherednik algebra if and only if the following relation holds for all $i=1, \dots, l$:
    \begin{equation*}
        \sum\limits_{\alpha \in R \cap V_i} \frac{c_{\alpha} (\alpha , u) (\alpha, v)}{(\alpha, \alpha)} = (u, v)
    \end{equation*}
    for all $u,v \in V_i$, where $V_i$ is the linear space spanned by the roots $\Gamma_i^v$.
\end{theorem}

The proof is similar to the proof of Theorem~\ref{invariance}.

Assume that $\mathcal{I}_{\Gamma_0}$ is invariant under the trigonometric Cherednik algebra. Then the space $\mathcal{C}_{W x_0}(U)$ can be identified with the quotient module $\mathcal{C}_{W x_0}(V, U) / \mathcal{I}_{\Gamma_0}$. 
The generalised trigonometric spin Calogero--Moser Hamiltonians are then defined similarly to the rational case by
\begin{equation*}
    H_2 = \widetilde{\Res}_\pi\left(\sum\limits_{i=1}^N (\nabla^{\rm tr}_{e_i})^2\right).
\end{equation*} 
Higher integrals are defined by
\begin{equation}
\label{trighigherham}
    H_p = \widetilde{\rm Res}_{\pi} p(\nabla^{\rm tr}) \qquad 
    (p \in \mathbb{C}[x]^{W}).
\end{equation}
Note that here $p(\nabla^{\rm tr})$ is $W$-invariant by~\cite{Ch2} (see also~\cite{HeckmanSurvey}).

To compute the Hamiltonians, we use an explicit expression for the sum of squares of trigonometric Dunkl operators:
\begin{equation*}
    \sum\limits_{i = 1}^N (\nabla^{\rm tr}_{e_i})^2 = \sum\limits_{i = 1}^N \partial_{x_i}^2 - \sum\limits_{\alpha \in R_+} c_{\alpha} \coth\left(\frac{(\alpha,x)}{2}\right)\partial_{\alpha} + \sum\limits_{\alpha \in R_+} \frac{c_{\alpha}(\alpha, \alpha)}{4 \sinh^2\left(\frac{(\alpha, x)}{2}\right)}(1 - s_{\alpha}) + (\rho, \rho).
\end{equation*}
We thus get the following expression for the generalised trigonometric spin Calogero--Moser operator
\begin{equation}
\label{trigham}
    H_2 = \Delta_y - \sum\limits_{\substack{\alpha \in R_+ \\ \widehat{\alpha} \neq 0}} c_{\alpha} \coth\left(\frac{(\widehat{\alpha}, y)}{2}\right) \partial_{\widehat{\alpha}} + \sum\limits_{\substack{\alpha \in R_+ \\ \widehat{\alpha} \neq 0}} \frac{c_{\alpha} (\alpha, \alpha)}{4 \sinh^2 \left(\frac{(\widehat{\alpha},y)}{2}\right)}(1 - P_{\alpha}) + (\rho, \rho),
\end{equation}
where we use the same notations as in the rational case. 
To bring the Hamiltonian $H_2$ into the potential-gauge form, we proceed similarly as before.

\begin{theorem}\label{trigpot}
Let $\pi \subset V$ be an intersection of mirrors
\begin{equation*}
    \pi = \{x \in V \mid (\beta, x) = 0, \ \forall \beta \in \Gamma^v_0 \}
\end{equation*}
corresponding to the Coxeter subgraph $\Gamma_0 \subset \Gamma$. 
Define the generalised coupling constants as in the rational case
\begin{equation*}
    \widehat{c}_{\widehat{\alpha}} = \sum\limits_{\substack{\gamma \in R_{+} \\ \widehat{\gamma} = \widehat{\alpha}}} c_{\gamma}.
\end{equation*}
Then
\begin{align*}
    f^{-1} &H_2 f = \Delta_y - \sum\limits_{\widehat{\alpha} \in \widehat{R}_+ \setminus \{0\} }\frac{\widehat{c}_{\widehat{\alpha}} (\widehat{\alpha},\widehat{\alpha})\widehat{P}_{\widehat{\alpha}}}{4 \sinh^2 \left(\frac{(\widehat{\alpha},y)}{2}\right)} 
    \\
    &- \sum\limits_{\widehat{\alpha} \in \widehat{R}_+ \setminus \{0\} } \sum\limits_{\substack{\widehat\beta \in \widehat{R}_+ \setminus \{0\} \\\widehat{\beta} \sim \widehat{\alpha}}}\frac{\widehat{c}_{\widehat{\alpha}} \widehat{c}_{\widehat{\beta}}(\widehat{\alpha}, \widehat{\beta})}{4} \coth\left(\frac{(\widehat{\alpha},y)}{2}\right) \coth\left(\frac{(\widehat{\beta},y)}{2}\right)  + \lambda,
\end{align*}
where 
\begin{equation*}
    f = \prod_{\widehat\alpha \in \widehat{R}_{+} \setminus \{0\}}\left( \sinh\left(\frac{(\widehat{\alpha}, y)}{2}\right)\right)^{\widehat{c}_{\widehat{\alpha}}}, 
\end{equation*}
\begin{equation*}
    \lambda = (\rho, \rho) - \frac14 \sum_{\alpha \in R_+}
    \sum_{\substack{\beta \in R_+ \\  \widehat{\beta} \nsim \widehat{\alpha}}} c_{\alpha} c_{\beta} (\widehat{\alpha}, \widehat{\beta}) 
    = (\rho, \rho) - \frac14 \sum_{\widehat{\alpha} \in \widehat{R}_+}
    \sum_{\substack{\widehat{\beta} \in \widehat{R}_+ \\  \widehat{\beta} \nsim \widehat{\alpha}}} \widehat{c}_{\widehat{\alpha}} \widehat{c}_{\widehat{\beta}} (\widehat{\alpha}, \widehat{\beta}).
\end{equation*}
\end{theorem}

\begin{proof}
We compute that
\begin{align}
    f^{-1} &H_2 f = \Delta_y - \sum\limits_{\substack{\alpha \in R_+ \\ \widehat{\alpha} \neq 0}}\frac{c_{\alpha}\big((\widehat{\alpha}, \widehat{\alpha}) + (\alpha, \alpha) (P_{\alpha}-1)\big)}{4 \sinh^2 \left(\frac{(\widehat{\alpha}, y)}{2}\right)} \label{form1}
    \\
     & - \frac{1}{4} \sum\limits_{\substack{\alpha \in R_+ \\ \widehat{\alpha} \neq 0}} \sum\limits_{\substack{ \beta \in R_+ \\ \widehat{\beta} \neq 0 }} c_{\alpha} c_{\beta} (\widehat{\alpha}, \widehat{\beta}) \coth\left(\frac{(\widehat{\alpha}, y)}{2}\right) \coth\left(\frac{(\widehat{\beta}, y)}{2}\right) + (\rho, \rho). \nonumber
\end{align}
Next, we simplify the last sum in \eqref{form1}. We will use a trigonometric analogue of \cite[Equality (10)]{Feigin}.
Namely, we claim for all $\alpha \in R_+$ that
\begin{equation}
\label{trigeq}
    \sum\limits_{\substack{\beta \in R_+ \\ \widehat{\beta} \nsim \widehat{\alpha}}} 
    c_{\beta} (\widehat{\alpha},\widehat{\beta}) \coth\left(\frac{(\widehat{\beta}, y)}{2}\right) = 0
\end{equation}
for $y \in \pi$ with\footnote{We use $\boldsymbol{\pi}$ for the numerical constant and $\pi$ for planes. Likewise, we use $\boldsymbol{i}$ for the imaginary unit to avoid any confusion with the indices.} $(\widehat{\alpha}, y) = 2  \boldsymbol{\pi} \boldsymbol{i}k$ for $k \in \mathbb{Z}$, which can be seen as follows. 

Let $W_{\alpha} = \langle W_0 , s_{\alpha} \rangle$ be the group generated by $s_{\alpha}$ and reflections about the hyperplanes orthogonal to the roots in $\Gamma_0^v$. Let $S \subset R$ be the set of the roots $\beta \in R$ such that $\widehat{\beta}$ is not proportional to $\widehat{\alpha}$. Decompose $S$ into $W_{\alpha}$-orbits $S = \mathcal{O}_1 \sqcup \dots \sqcup \mathcal{O}_r $.
We will show that
\begin{equation}
    \label{trigeq2}
    \sum\limits_{\beta \in \mathcal{O}_i} c_{\beta} (\widehat{\alpha}, \widehat{\beta}) \coth\left(\frac{(\widehat{\beta}, y)}{2}\right) = 0
\end{equation}
for all $i$. Let $\beta_1, \beta_2 \in \mathcal{O}_i$. Then $\coth(\frac{(\widehat{\beta}_1, y)}{2}) = \coth(\frac{(\widehat{\beta}_2, y)}{2})$. Indeed, this is evident if $\beta_1 = s_0 \beta_2$ for $s_0 \in W_0$; and if $\beta_1 = s_{\alpha} \beta_2$, then
\begin{equation*}
\begin{array}{c}
\displaystyle{
    \coth\left(\frac{(\widehat{s_{\alpha} \beta_2}, y)}{2}\right) = \coth\left(\frac{(s_{\alpha} \beta_2 , y)}{2}\right) = \coth\left(\frac{ (\beta_2 , s_{\alpha} y)}{2}\right) 
    }
    \\ \ \\
    \displaystyle{
     = \coth \left(\frac{(\widehat{\beta_2}, y)}{2} - \boldsymbol{\pi} \boldsymbol{i} k \frac{2 (\beta_2 , \alpha)}{(\alpha, \alpha)} \right) 
     = \coth \left(\frac{(\widehat{\beta_2}, y)}{2}\right),
    }
    \end{array}
\end{equation*}
where we used that $\frac{2 (\beta_2 , \alpha)}{(\alpha, \alpha)} \in \mathbb{Z}$. Now let
$
    b_i = \sum_{\beta \in \mathcal{O}_i} \beta.
$
It satisfies the relations $s_0 b_i = b_i$ for $s_0 \in W_0$ and $s_{\alpha} b_i = b_i$. This translates into $\widehat{b}_i = b_i$ and $0= (\alpha, b_i) = (\widehat{\alpha}, \widehat{b}_i)$, which implies \eqref{trigeq2} and hence \eqref{trigeq}. 

Relation \eqref{trigeq} implies that the expression
\begin{equation*}
   \sum_{\alpha \in R_+}
    \sum_{\substack{\beta \in R_+ \\  \widehat{\beta} \nsim \widehat{\alpha}}} c_{\alpha} c_{\beta} (\widehat{\alpha}, \widehat{\beta}) \coth\left(\frac{(\widehat{\alpha}, y)}{2}\right) \coth\left(\frac{(\widehat{\beta}, y)}{2}\right)
\end{equation*}
has no poles; thus it is an entire bounded function that is constant due to Liouville's theorem. The constant can be calculated to be $4 (\rho, \rho) - 4\lambda$ by a limit at infinity in a suitable chamber such that $\coth (\widehat{\alpha}, y) \to 1$ for all $\widehat{\alpha}$.
Using this fact to simplify expression~\eqref{form1}, the proof is then completed by using~\eqref{P hat} and the definition of $\widehat{c}_{\widehat{\alpha}}$.
\end{proof}

\begin{remark}
Let us assume that any collinear vectors in $\widehat{R}_+ \setminus \{0\}$ are of the form $\widehat{\alpha}$, $2 \widehat{\alpha}$. This is the case for the projections of all classical root systems. Then the operator $f^{-1}H_2 f$ from Theorem~\ref{trigpot} becomes (up to a constant, namely, $(\rho, \rho) - (\widehat{\rho}, \widehat{\rho})$) equal to
\begin{equation}
\label{simpleform}
    L \coloneqq \Delta_y - \sum\limits_{\widehat{\alpha} \in \widehat{R}_+ \setminus \{0\}}
    \frac{\widehat{c}_{\widehat{\alpha}} (\widehat{c}_{\widehat{\alpha}} + 2 \widehat{c}_{2\widehat{\alpha}} + \widehat{P}_{\widehat{\alpha}})(\widehat{\alpha},\widehat{\alpha})}{4 \sinh^2 \left(\frac{(\widehat{\alpha},y)}{2}\right)},
\end{equation}
where $\widehat{c}_{2\widehat{\alpha}} \coloneqq 0$ when $2\widehat{\alpha} \notin \widehat{R}_+$.
However, the above assumption may fail for exceptional root systems, for example, for $G_2$ and $F_4$ in some cases. 
\end{remark}

\subsection{Restricted operators from classical root systems}
In this section, we discuss the generalised trigonometric spin Calogero--Moser Hamiltonians in the case where $R$ is a classical root system. 
The resulting Hamiltonians depend only on the projected root system itself, except potentially for the definition of the operators $\widehat{P}_{\widehat{\alpha}}$. Though, for the reflection representation, the latter operators stand just for reflections about the hyperplanes orthogonal to projected roots. 
We have described projections of the root systems of types $A$, $B$, and~$D$ in the rational case, so we discuss these types only briefly in this section. However, in the trigonometric case, we also need to consider type $C$, which will be discussed in more detail.

\subsubsection{$A$-series}\label{trig type A}
Here we construct a deformed version of trigonometric spin Calogero--Moser Hamiltonians using the trigonometric Cherednik
algebra of type $GL_N$. The Hamiltonians of the \emph{non-deformed} trigonometric spin Calogero--Moser system are obtained by restricting symmetric functions of the Dunkl operators~\eqref{trigDunkl} for $R_+ = A_{N-1}^+$, $\xi = e_i$ ($i=1, \dots, N$), to the elements of the polynomial representation tensored with $U$ that are symmetric under simultaneous permutation in variables $x$ and spin variables.

Similarly to Section~\ref{rat type A}, for $c = \frac{1}{k}$, the trigonometric Cherednik algebra for $GL_N$ admits parabolic submodules $\mathcal{I}_{\Gamma_0}$ for the strata \eqref{strata}.  
Using formula~\eqref{simpleform}, we get that the corresponding deformed version of the trigonometric spin Calogero--Moser Hamiltonian is
\begin{align*}
    L &= 
    \sum\limits_{i = 1}^{m+n} \partial_{y_i}^2  
    - \sum\limits_{m+1 \leq  i < j \leq m+n} \frac{1 + k\widetilde{P}_{ij}}{2k^2\sinh^2(\frac{y_i - y_j}{2})} \\
    &\ - \sum\limits_{i = 1}^m\sum\limits_{j = m+1}^{m+n}
     \frac{1 + \widetilde{P}_{ij}}{2k\sinh^2\left(\frac{y_i - \sqrt{k}y_j}{2\sqrt{k}}\right)}  -
    \sum\limits_{1 \leq i < j \leq m} \frac{k + \widetilde{P}_{ij}}{2k\sinh^2\left(\frac{y_i - y_j}{2\sqrt{k}}\right)},
\end{align*}
where we use the same notations as in Proposition~\ref{rat A prop}.

\subsubsection{$B$-series}
The most general projection of the root system of type~$B$ to consider is \eqref{projroot} for which, by using formula \eqref{simpleform}, the generalised trigonometric spin Calogero--Moser operator is
\begin{align}
    L &= \sum_{i=1}^{m+n} \partial_{y_i}^2 -  \sum\limits_{m+1 \leq i < j \leq m+n} \frac{c_1}{2} \left( \frac{c_1 + \widehat{P}_{\widetilde{e}_i - \widetilde{e}_j}}{\sinh^2(\frac{y_i - y_j}{2})} + \frac{c_1 + \widehat{P}_{\widetilde{e}_i + \widetilde{e}_j}}{\sinh^2(\frac{y_i + y_j}{2})} \right) \label{trigHamB}
    \\
    &\quad -\sum\limits_{1 \leq i < j \leq m}\frac{c_1 k}{2}  \left( \frac{c_1 k^2 + \widehat{P}_{\frac{\widetilde{e}_i - \widetilde{e}_j}{\sqrt{k}}}}{\sinh^2(\frac{y_i - y_j}{2\sqrt{k}})} + \frac{c_1 k^2 + \widehat{P}_{\frac{\widetilde{e}_i + \widetilde{e}_j}{\sqrt{k}}}}{\sinh^2(\frac{y_i + y_j}{2\sqrt{k}})} \right) 
    \nonumber \\
    &\quad - \sum\limits_{i = 1}^m \sum\limits_{j = m+1}^{m+n} \frac{c_1(k+1)}{4}\left( \frac{c_1 k +  \widehat{P}_{\frac{\widetilde{e}_i}{\sqrt{k}}-\widetilde{e}_j}}{\sinh^2(\frac{y_i-\sqrt{k}y_j}{2\sqrt{k}})} + \frac{c_1 k + \widehat{P}_{\frac{\widetilde{e}_i}{\sqrt{k}}+\widetilde{e}_j}}{\sinh^2(\frac{y_i + \sqrt{k}y_j}{2\sqrt{k}})} \right) 
    \nonumber \\
    &\quad -\sum\limits_{i = m+1}^{m+n} \frac{(2c_1 l + c_2)^2 + (2 c_1 l + c_2)\widehat{P}_{\widetilde{e}_i}}{4\sinh^2(\frac{y_i}{2})} 
    - \sum\limits_{i = 1}^m \frac{(k-1)^2 + 2(k-1) \widehat{P}_{\frac{2 \widetilde{e}_i}{\sqrt{k}}}}{4k\sinh^2(\frac{y_i}{\sqrt{k}})} 
    \nonumber \\
    &\quad -\sum\limits_{i = 1}^m \frac{(2 c_1 l + c_2)\left(2 c_1 kl +  c_2 k + k-1 + \widehat{P}_{\frac{\widetilde{e}_i}{\sqrt{k}}}\right)}{4\sinh^2(\frac{y_i}{2\sqrt{k}})}. \nonumber
\end{align}

\subsubsection{$C$-series}
A positive half of the root system $C_N$ is $C_N^+ = \{ e_i \pm e_j \mid 1 \leq i < j \leq N\} \cup \{ 2e_i  \mid i= 1, \dots, N\}$.
Let $c_1$ be the multiplicity of the roots $e_i \pm e_j$ and $c_2$ the multiplicity of the roots $2e_i$.

The most general parabolic strata leading to an invariant parabolic submodule in this case correspond to $\pi_{m,k;\, l} \subset \mathbb{C}^N$ for $mk+l \leq N$ defined by the equations
\begin{align*}
    &x_1 = x_2 = \dots = x_k, \ x_{k+1} = x_{k + 2} = \dots = x_{2k}, \ \dots, \\
    &x_{(m-1)k+1} = x_{(m-1)k+2} = \dots = x_{mk}, \\
    &x_{N - l + 1} = x_{N - l + 2} = \dots = x_{N} = 0,    
\end{align*}
where we allow $k=1$ or $l=0$, which means that the corresponding relations are absent.

For the corresponding $\mathcal{I}_{\Gamma_0}$ to be invariant,
the multiplicities must satisfy the following. When $k > 1$ and $l \neq 0$, we must have $c_1 = \frac{1}{k}$ and $c_2 = \frac{1}{2} - \frac{l - 1}{k}$. When $k > 1$ and $l = 0$, we must have $c_1 = \frac{1}{k}$ and there is no restriction on $c_2$. When $k = 1$ and $l \neq 0$, we must have $2c_1(l-1) + 2c_2 = 1$.


The corresponding projected root system (the non-zero vectors in the projection of $C_N^+$ onto~$\pi_{m,k; \, l}$) is 
\begin{align*}
    C^+_{m,n; \, k} &= \{ \widetilde{e}_i \pm \widetilde{e}_j \mid m+1 \leq i < j \leq m+n\} \\
    &\quad \cup \{\tfrac{\widetilde{e}_i}{\sqrt{k}} \pm \widetilde{e}_j \mid 1 \leq i \leq m, \ m+1 \leq j \leq m+n \}  
    \\
    &\quad \cup\{\tfrac{\widetilde{e}_i \pm \widetilde{e}_j}{\sqrt{k}} \mid 1 \leq i < j \leq m \} \cup 
    \{\widetilde{e}_i \mid m+1 \leq i \leq m+n \}
    \\
    &\quad \cup \{2\widetilde{e}_i \mid m+1 \leq i \leq m+n\} \cup \{\tfrac{\widetilde{e}_i}{\sqrt{k}} \mid 1 \leq i \leq m \} \\
    &\quad \cup \{\tfrac{2 \widetilde{e}_i}{\sqrt{k}} \mid 1 \leq i \leq m \},
\end{align*}
where $n = N - mk - l$, the vectors $\widetilde{e}_i$ are defined as in Section~\ref{rat type B}, and the respective generalised multiplicities are 
\begin{alignat*}{4}
    &\widehat{c}_{\widetilde{e}_i \pm \widetilde{e}_j} = c_1, \qquad 
    &&\widehat{c}_{(\sqrt{k})^{-1}\widetilde{e}_i \pm \widetilde{e}_j} = c_1 k, \qquad 
    &&\widehat{c}_{(\sqrt{k})^{-1}(\widetilde{e}_i \pm \widetilde{e}_j)} = c_1 k^2, \qquad 
    &&\widehat{c}_{\widetilde{e}_i} = 2c_1 l , \\
    &\widehat{c}_{2\widetilde{e}_i} = c_2, \qquad 
    &&\widehat{c}_{(\sqrt{k})^{-1}\widetilde{e}_i} = 2c_1 kl, \qquad
    &&\widehat{c}_{2(\sqrt{k})^{-1}\widetilde{e}_i} = \frac{k-1}{2} + c_2k. 
\end{alignat*}
Any vector with zero multiplicity is thought of as excluded from $C_{m,n; \, k}^+$.
This projection leads to a case of either the undeformed or deformed $BC$-type root system. 

Using formula~\eqref{simpleform}, the Hamiltonian corresponding to $\pi_{m,k; \, l}$ is
\begin{align*}
    L &= \sum_{i=1}^{m+n} \partial_{y_i}^2 
    - \sum\limits_{m+1 \leq i < j \leq m+n} \frac{c_1}{2} \left(\frac{c_1 + \widehat{P}_{\widetilde{e}_i - \widetilde{e}_j}}{\sinh^2 (\frac{y_i - y_j}{2})} + \frac{c_1 + \widehat{P}_{\widetilde{e}_i + \widetilde{e}_j}}{\sinh^2 (\frac{y_i + y_j}{2})} \right) 
    \\
    &\quad - \sum\limits_{1 \leq i < j \leq m} \frac{c_1 k}{2} \left(\frac{c_1k^2 +  \widehat{P}_{\frac{\widetilde{e}_i - \widetilde{e}_j}{\sqrt{k}}}}{\sinh^2(\frac{y_i - y_j}{2\sqrt{k}})} + \frac{c_1k^2 +  \widehat{P}_{\frac{\widetilde{e}_i + \widetilde{e}_j}{\sqrt{k}}}}{\sinh^2(\frac{y_i + y_j}{2\sqrt{k}})} \right) 
    \\
    &\quad - \sum\limits_{i = 1}^m \sum\limits_{j = m+1}^{m+n} \frac{c_1(k+1)}{4} \left( \frac{c_1 k +  \widehat{P}_{\frac{\widetilde{e}_i}{\sqrt{k}} -\widetilde{e}_j}}{\sinh^2(\frac{y_i - \sqrt{k} y_j}{2 \sqrt{k}})} + 
    \frac{c_1 k +  \widehat{P}_{\frac{\widetilde{e}_i}{\sqrt{k}} + \widetilde{e}_j}}{\sinh^2(\frac{y_i + \sqrt{k} y_j}{2 \sqrt{k}})}\right) 
    \\
    &\quad - \sum\limits_{i = m+1}^{m+n}  \frac{c_1l(2 c_1 l + 2 c_2  +  \widehat{P}_{\widetilde{e}_i})}{2 \sinh^2(\frac{y_i}{2})} 
    -   \sum\limits_{i = 1}^m \frac{ c_1 l \left( 2c_1 kl + k-1 + 2c_2 k + \widehat{P}_{\frac{\widetilde{e}_i}{\sqrt{k}}} \right)}{2 \sinh^2 (\frac{y_i}{2 \sqrt{k}})} 
    \\
    &\quad - \sum\limits_{i = m+1}^{m+n}\frac{c_2^2 + c_2 \widehat{P}_{2\widetilde{e}_i}}{\sinh^2(y_i)} 
    - \sum\limits_{i = 1}^m\frac{(k-1 + 2 c_2k) \left(k-1 + 2c_2k+2\widehat{P}_{\frac{2 \widetilde{e}_i}{\sqrt{k}}}\right)}{4k \sinh^2 (\frac{y_i}{\sqrt{k}})} ,
\end{align*}
where the variables $y_i$ are defined as in Section~\ref{rat type B}.

\subsubsection{$D$-series}
In this case, the Hamiltonians that we get are of the form \eqref{trigHamB} with multiplicities $c_1$ and $c_2$ as described in Section~\ref{D series} for the strata $\pi_{m,k}^+$, $\pi_p$, and $\pi_{m,k}$, respectively, and $\widehat{P}_\alpha$ defined by formula~\eqref{P hat}.

\subsection{Examples for exceptional root systems}
In this section, we consider all two-dimensional
projections of the exceptional root systems of types $E$ and~$F_4$, one higher-dimensional projection of $E_7$, and we comment on a one-dimensional projection of $G_2$.

\subsubsection{Type $E$}

\hfill \newline

\noindent (i) $(E_8 , A_3 \times A_3)$  \\
The root system $E_8 \subset \mathbb{R}^8$ has simple roots

\begin{equation}
\label{E_8 root}
\begin{array}{c}
\displaystyle{
    \alpha_1 = \frac{1}{2}(e_1 - e_2 - e_3 - e_4 - e_5 - e_6 - e_7+ e_8), \ \alpha_2 = e_1 + e_2,  
    }
    \\ \ \\
    \displaystyle{
    \alpha_3 = e_2 - e_1, \ \alpha_4 = e_3 - e_2, \ \alpha_5 = e_4 - e_3, \ \alpha_6 = e_5 - e_4, 
    }
    \\ \ \\
    \displaystyle{
    \alpha_7 = e_6 - e_5, \ \alpha_8 = e_7 - e_6,
    }
\end{array}
\end{equation}
where $\{e_i\}_{i=1}^8$ is an orthonormal basis in $\mathbb{R}^8$. Its Dynkin diagram is
\begin{equation*}
\begin{tikzpicture}
\filldraw [red] (0,0) circle (3pt) node[below right, black] {$\alpha_4$} ;
\filldraw [red] (-1.5,0) circle (3pt) node[below, black]  {$\alpha_3$};
\filldraw [black] (-3,0) circle (3pt) node[below,black] {$\alpha_1$};
\filldraw [black] (1.5,0) circle (3pt) node[below,black] {$\alpha_5$};
\filldraw [red] (3,0) circle (3pt) node[below,black] {$\alpha_6$};
\filldraw [red] (0,-1.5) circle (3pt) node[right,black] {$\alpha_2$};
\filldraw [red] (4.5,0) circle (3pt) node[below, black] {$\alpha_7$};
\filldraw [red] (6,0) circle (3pt) node[below, black] {$\alpha_8$};
\draw[thick] (-3,0) -- (-1.6,0);
\draw[thick][red] (-1.5,0) -- (0,0);
\draw[thick][red] (0,0) -- (0,-1.5);
\draw[thick][red] (3, 0) -- (6, 0);
\draw[thick][black] (0.1,0) -- (2.9, 0);

\end{tikzpicture}
\end{equation*}
where the red vertices and edges indicate the chosen subgraph $\Gamma_0 \simeq A_3 \times A_3$. The corresponding plane $\pi$ is defined by the equations
\begin{equation*}
    x_1 = x_2 = x_3 = 0, \ x_4 = x_5 = x_6 = x_7,
\end{equation*}
and the multiplicity must be $c = \frac{1}{4}$. The projected system is two-dimensional and is shown in the following diagram, where the coordinates are in the basis formed by $\widetilde{e}_1 = \frac12 (e_4 + e_5 + e_6 + e_7)$ and $\widetilde{e}_2 = e_8$:
\begin{figure}[H]
\includegraphics[scale=0.6]{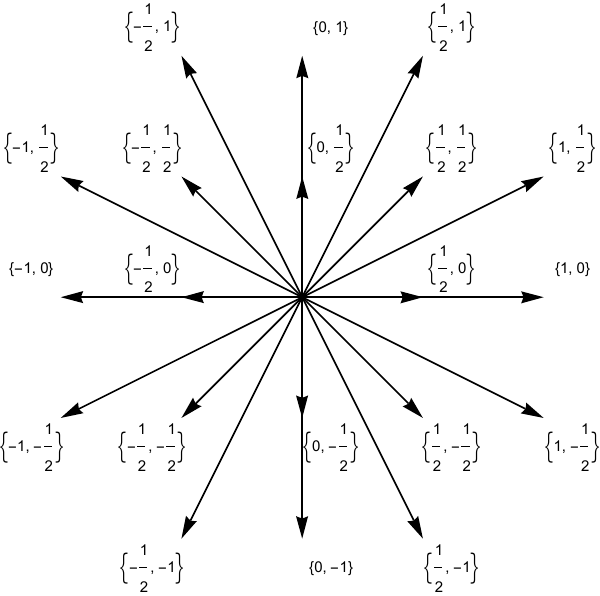}
\end{figure}
\noindent The multiplicities of the positive half are
\begin{equation*}
    \widehat{c}_{\frac{\widetilde{e}_i}{2}} = 6, \quad 
    \widehat{c}_{\widetilde{e}_i} = \frac32, \quad 
    \widehat{c}_{\pm \frac{\widetilde{e}_1}{2} + \frac{\widetilde{e}_2}{2}} = 4, \quad
    \widehat{c}_{\pm \widetilde{e}_1 + \frac{\widetilde{e}_2}{2}} = 1, \quad 
    \widehat{c}_{\pm \frac{\widetilde{e}_1}{2} + \widetilde{e}_2} = 1. 
\end{equation*}

By formula~\eqref{simpleform}, the corresponding generalised trigonometric spin Calogero--Moser Hamiltonian is
\begin{align*}
    L &= \partial_{y_1}^2 + \partial_{y_2}^2 - \sum\limits_{i = 1}^2 \left( \frac{27 + 3 \widehat{P}_{\frac{\widetilde{e}_i}{2} }}{8 \sinh^2(\frac{y_i}{4})} + \frac{9 + 6 \widehat{P}_{\widetilde{e}_i}}{16 \sinh^2(\frac{y_i}{2})} \right) - \frac{4 +  \widehat{P}_{\frac{\widetilde{e}_2 - \widetilde{e}_1}{2}}}{2 \sinh^2(\frac{y_2 - y_1}{4})} 
    \\
    &\quad - \frac{4 +  \widehat{P}_{\frac{\widetilde{e}_1 + \widetilde{e}_2}{2}}}{2 \sinh^2(\frac{y_1 + y_2}{4})} -
    \frac{5 + 5\widehat{P}_{ \frac{\widetilde{e}_2}{2}- \widetilde{e}_1}}{16 \sinh^2(\frac{y_2 - 2y_1}{4})} -
    \frac{5 + 5\widehat{P}_{\widetilde{e}_1 + \frac{\widetilde{e}_2}{2}}}{16 \sinh^2(\frac{2y_1 + y_2}{4})}  
    \\
    & \quad - \frac{5 + 5\widehat{P}_{\widetilde{e}_2 -  \frac{\widetilde{e}_1}{2}}}{16 \sinh^2(\frac{2y_2-y_1}{4})} - 
    \frac{5 + 5\widehat{P}_{ \frac{\widetilde{e}_1}{2} + \widetilde{e}_2}}{16 \sinh^2(\frac{y_1 + 2y_2}{4})},
\end{align*}
where $y_1 = \frac12 (x_4 + x_5 + x_6 + x_7)$ and $y_2 = x_8$. This operator is a trigonometric version with spin of the potential-gauge form of the operator~\cite[Formula (28)]{Feigin} with $m = \frac{15}{2}$, $n=4$, and $\alpha = (\sqrt{2n+1} + \sqrt{2(m+n+1)})/\sqrt{2m+1} = 2$.

The other possible choice of a subgraph $\Gamma_0 \simeq A_3 \times A_3$ in the Dynkin diagram of~$E_8$ leads to an equivalent projected configuration.

\hfill \newline
\noindent (ii) $(E_8 , A_6)$\\
As another example, let us choose a subgraph $A_6$ in the Dynkin diagram of $E_8$ as in the following picture 

\begin{equation*}
\begin{tikzpicture}
\filldraw [red] (0,0) circle (3pt) node[below right, black] {$\alpha_4$} ;
\filldraw [red] (-1.5,0) circle (3pt) node[below, black]  {$\alpha_3$};
\filldraw [black] (-3,0) circle (3pt) node[below,black] {$\alpha_1$};
\filldraw [red] (1.5,0) circle (3pt) node[below,black] {$\alpha_5$};
\filldraw [red] (3,0) circle (3pt) node[below,black] {$\alpha_6$};
\filldraw [black] (0,-1.5) circle (3pt) node[right,black] {$\alpha_2$};
\filldraw [red] (4.5,0) circle (3pt) node[below, black] {$\alpha_7$};
\filldraw [red] (6,0) circle (3pt) node[below, black] {$\alpha_8$};
\draw[thick] (-3,0) -- (-1.6,0);
\draw[thick][red] (-1.5,0) -- (0,0);
\draw[thick][black] (0,-0.1) -- (0,-1.5);
\draw[thick][red] (3, 0) -- (6, 0);
\draw[thick][red] (0, 0) -- (3, 0);

\end{tikzpicture}
\end{equation*}
The corresponding plane $\pi$ is defined by the equations
\begin{equation*}
    x_1 = x_2 = \dots = x_7, 
\end{equation*}
with multiplicity $c = \frac{1}{7}$. The corresponding projected system is shown in the following diagram, where the coordinates are in the basis formed by $\widetilde{e}_1 = \frac{1}{\sqrt{7}}(e_1 + \dots + e_7)$ and $\widetilde{e}_2 = e_8$:

\begin{figure}[H]
\includegraphics[scale = 0.6]{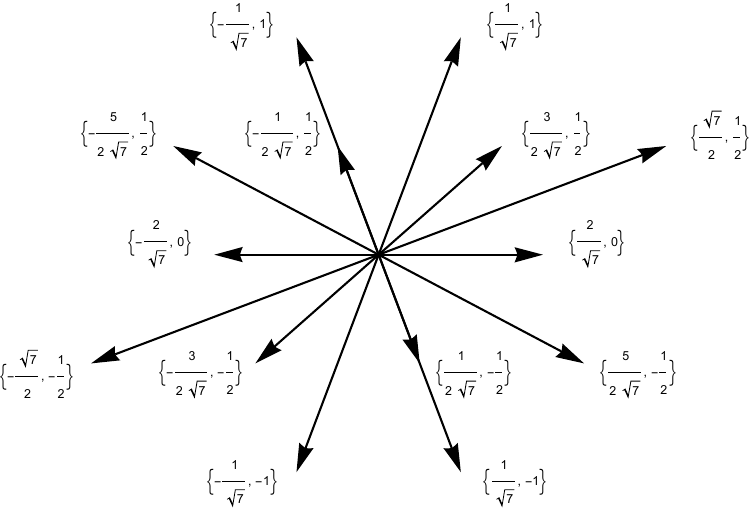}
\end{figure}
This configuration resembles the configurations $G_2$ and $AG_2$, but it has exactly one line $\ell$ containing collinear vectors, and the configuration is scaled in the orthogonal direction to $\ell$ compared to $G_2$ and $AG_2$. The multiplicities are given by
\begin{equation*} 
    \widehat{c}_{\frac{\widetilde{e}_2}{2} - \frac{5 \widetilde{e}_1}{2\sqrt{7}}} =
    \widehat{c}_{\widetilde{e}_2 \pm \frac{\widetilde{e}_1}{\sqrt{7}}} = 1, \quad
    \widehat{c}_{\frac{\widetilde{e}_2}{2}-\frac{\widetilde{e}_1}{2\sqrt{7}}} = 5, \quad
    \widehat{c}_{\frac{3\widetilde{e}_1}{2\sqrt{7}} + \frac{\widetilde{e}_2}{2}} = 
    \widehat{c}_{\frac{2\widetilde{e}_1}{\sqrt{7}}} = 3, \quad
    \widehat{c}_{\frac{\sqrt{7}\widetilde{e}_1}{2} + \frac{\widetilde{e}_2}{2}} = \frac17. 
\end{equation*}

By formula~\eqref{simpleform}, the corresponding generalised trigonometric spin Calogero--Moser Hamiltonian is
\begin{align*}
    L &= \partial_{y_1}^2 + \partial_{y_2}^2 - \frac{35 + 5 \widehat{P}_{\frac{\sqrt{7}\widetilde{e}_2 - \widetilde{e}_1}{2\sqrt{7}}}}{14 \sinh^2(\frac{\sqrt{7}y_2 - y_1}{4 \sqrt{7}})} - \frac{2 + 2 \widehat{P}_{\widetilde{e}_2 - \frac{\widetilde{e}_1}{\sqrt{7}}}}{7 \sinh^2(\frac{\sqrt{7}y_2 - y_1}{2\sqrt{7}})}
    \\
    &\quad - \frac{2 + 2 \widehat{P}_{\frac{\widetilde{e}_1}{\sqrt{7}} + \widetilde{e}_2}}{7 \sinh^2(\frac{y_1+\sqrt{7} y_2}{2 \sqrt{7}})} - \frac{2 + 2\widehat{P}_{\frac{\sqrt{7}\widetilde{e}_2 - 5 \widetilde{e}_1}{2 \sqrt{7}}}}{7 \sinh^2(\frac{\sqrt{7}y_2 - 5 y_1}{4 \sqrt{7}})} - \frac{9 + 3 \widehat{P}_{\frac{2\widetilde{e}_1}{\sqrt{7}}}}{7 \sinh^2(\frac{y_1}{\sqrt{7}})}
    \\
    &\quad - \frac{1 + 7 \widehat{P}_{\frac{\sqrt{7}\widetilde{e}_1 + \widetilde{e}_2}{2}}}{98 \sinh^2(\frac{\sqrt{7}y_1 + y_2}{4})} - 
    \frac{9 + 3 \widehat{P}_{\frac{3 \widetilde{e}_1 + \sqrt{7}\widetilde{e}_2}{2 \sqrt{7}}}}{7 \sinh^2(\frac{3 y_1 + \sqrt{7}y_2}{4 \sqrt{7}})},
\end{align*}
where $y_1 = \frac{1}{\sqrt{7}}(x_1 + \dots + x_7)$ and $y_2 = x_8$.

Other choices of a subgraph $\Gamma_0 \simeq A_6$ in the Dynkin diagram of~$E_8$ lead to equivalent projected configurations.
The subgraph $\Gamma_0 \simeq E_6$ leads to $G_2$ as in table~\cite[p.~272]{Feigin} (the multiplicity $\frac{23}{12}$ should be $\frac{27}{12}$). And $\Gamma_0 \simeq D_6$ leads to $BC_2$ with multiplicities $\frac{16}{5}$, $\frac{6}{5}$, and $\frac{1}{10}$ for the vectors $e_i$, $e_i \pm e_j$, and $2e_i$, respectively (cf.~\cite[p.~272]{Feigin}).

\hfill \newline
\noindent (iii) $(E_7, D_5)$ \\
The Dynkin diagram of the root system $E_7$ is

\begin{equation*}
\begin{tikzpicture}
\filldraw [red] (0,0) circle (3pt) node[below right, black] {$\alpha_4$} ;
\filldraw [red] (-1.5,0) circle (3pt) node[below, black]  {$\alpha_3$};
\filldraw [black] (-3,0) circle (3pt) node[below,black] {$\alpha_1$};
\filldraw [red] (1.5,0) circle (3pt) node[below,black] {$\alpha_5$};
\filldraw [red] (3,0) circle (3pt) node[below,black] {$\alpha_6$};
\filldraw [red] (0,-1.5) circle (3pt) node[right,black] {$\alpha_2$};
\filldraw [black] (4.5,0) circle (3pt) node[below,black] {$\alpha_7$};
\draw[thick] (-3,0) -- (-1.6,0);
\draw[thick][red] (-1.5,0) -- (3,0);
\draw[thick][red] (0,0) -- (0,-1.5);
\draw[thick][black] (3.1,0) -- (4.5, 0);
\end{tikzpicture}
\end{equation*}
with the simple roots being the first seven roots of \eqref{E_8 root}. In this example, we choose a subgraph $\Gamma_0 \simeq D_5$. The corresponding subspace $\pi$ is given by the equations
\begin{equation*}
    x_1 = x_2 = x_3 = x_4 = x_5 = 0,
\end{equation*}
with multiplicity $c = \frac{1}{8}$.
The corresponding projected system is shown in the following diagram, where the coordinates are with respect to $\widetilde{e}_1 = e_6$ and $\widetilde{e}_2 = \frac{1}{\sqrt{2}}(e_8 - e_7)$:
\begin{figure}[H]
\includegraphics[scale = 0.45]{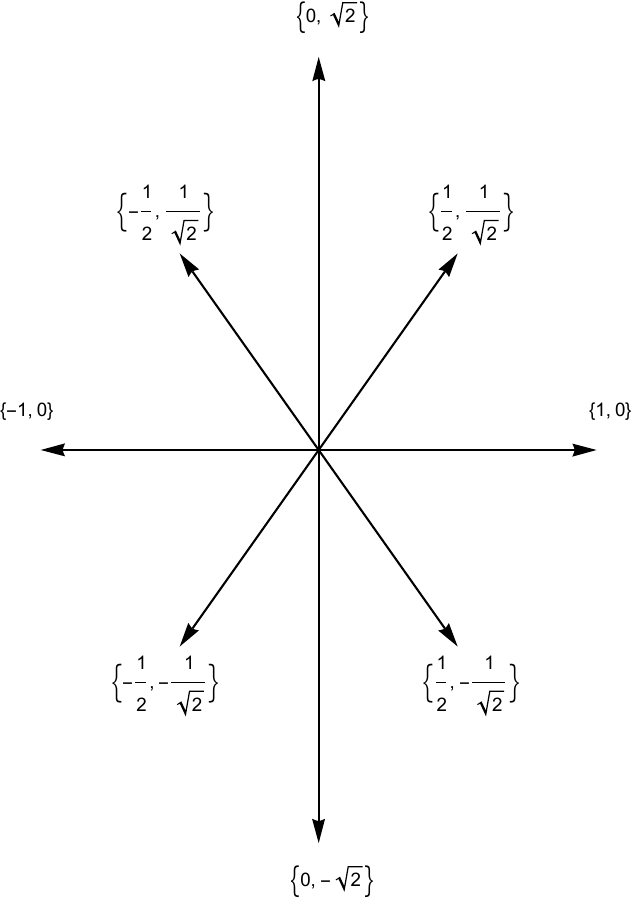}
\end{figure}
\noindent The multiplicities are given by
\begin{equation*}
    \widehat{c}_{\widetilde{e}_1} = \frac54, \quad \widehat{c}_{\sqrt{2}\widetilde{e}_2} = \frac18, \quad 
    \widehat{c}_{\frac{\widetilde{e}_2}{\sqrt{2}} \pm \frac{\widetilde{e}_1}{2}} = 2.
\end{equation*}
These vectors are as in a deformed $C_2$ configuration 
except for having different multiplicities.

By formula~\eqref{simpleform}, the corresponding generalised trigonometric spin Calogero--Moser Hamiltonian is
\begin{equation*}
\begin{array}{c}
\displaystyle{
    L = \partial_{y_1}^2 + \partial_{y_2}^2 - 
    \frac{25 + 20 \widehat{P}_{\widetilde{e}_1}}{64 \sinh^2(\frac{y_1}{2})} -
    \frac{6 + 3 \widehat{P}_{\frac{\widetilde{e}_1 + \sqrt{2} \widetilde{e}_2}{2}}}{8\sinh^2(\frac{y_1 + \sqrt{2}y_2}{4})} 
    }
    \\ 
    \displaystyle{
     - \frac{6 + 3 \widehat{P}_{\frac{\sqrt{2} \widetilde{e}_2 - \widetilde{e}_1}{2}}}{8\sinh^2(\frac{\sqrt{2}y_2 - y_1}{4})}
     - \frac{1 + 8 \widehat{P}_{\sqrt{2} \widetilde{e}_2}}{128\sinh^2(\frac{y_2}{\sqrt{2}})},
    }
\end{array}
\end{equation*}
where $y_1 = x_6$ and $y_2 = \frac{1}{\sqrt{2}}(x_8 - x_7)$.

The scalar version of this operator, obtained by replacing all the occurrences of $\widehat{P}$ by the identity, coincides with the trigonometric degeneration (where the Weierstrass $\wp$ function is replaced by $\sinh^{-2}$) of the operator from \cite[Theorem 1.4]{Taniguchi} for $a=\sqrt{2}$. 

The other possible choice of a subgraph $\Gamma_0 \simeq D_5$ in the Dynkin diagram of~$E_7$ leads to an equivalent projected configuration. Other allowed two-dimensional projections of $E_7$ are as in the table \cite[p.~273]{Feigin}, except that here $(E_7, A_5)_1$ leads to the deformed $BC_2$ with deformation parameter $k = \frac43$ and multiplicities $\frac{10}{3}$, $\frac16$, $\frac52$, $0$, and $1$ for the vectors $e_1$, $2e_1$, $\sqrt{k}e_2$, $2\sqrt{k}e_2$, and $e_1 \pm \sqrt{k}e_2$, respectively.

\hfill \newline
\noindent (iv) $(E_6, A_4)$ \\
The Dynkin diagram of the root system $E_6$ is
\begin{equation*}
\begin{tikzpicture}[scale=1.1, transform shape]
\filldraw [red] (0,0) circle (3pt) node[below right, black] {$\alpha_4$} ;
\filldraw [red] (-1.5,0) circle (3pt) node[below, black]  {$\alpha_3$};
\filldraw [black] (-3,0) circle (3pt) node[below,black] {$\alpha_1$};
\filldraw [red] (1.5,0) circle (3pt) node[below,black] {$\alpha_5$};
\filldraw [red] (3,0) circle (3pt) node[below,black] {$\alpha_6$};
\filldraw [black] (0,-1.5) circle (3pt) node[right,black] {$\alpha_2$};
\draw[thick] (-3,0) -- (-1.6,0);
\draw[thick][red] (-1.5,0) -- (3,0);
\draw[thick] (0,-0.1) -- (0,-1.5);

\end{tikzpicture}
\end{equation*}
with the simple roots being the first six roots of \eqref{E_8 root}. Here we chose a subgraph $\Gamma_0 \simeq A_4$. The corresponding subspace $\pi$ is given by the equations
\begin{equation*}
    x_1 = x_2 = x_3 = x_4 = x_5,
\end{equation*}
with multiplicity $c = \frac{1}{5}$.
The corresponding projected system is shown in the following diagram, where the coordinates are with respect to $\widetilde{e}_1 = \frac{1}{\sqrt{3}}(e_8 - e_7 - e_6)$ and $\widetilde{e}_2 = \frac{1}{\sqrt{5}}(e_1 + \dots + e_5)$:
\begin{figure}[H]
\includegraphics[scale=0.50]{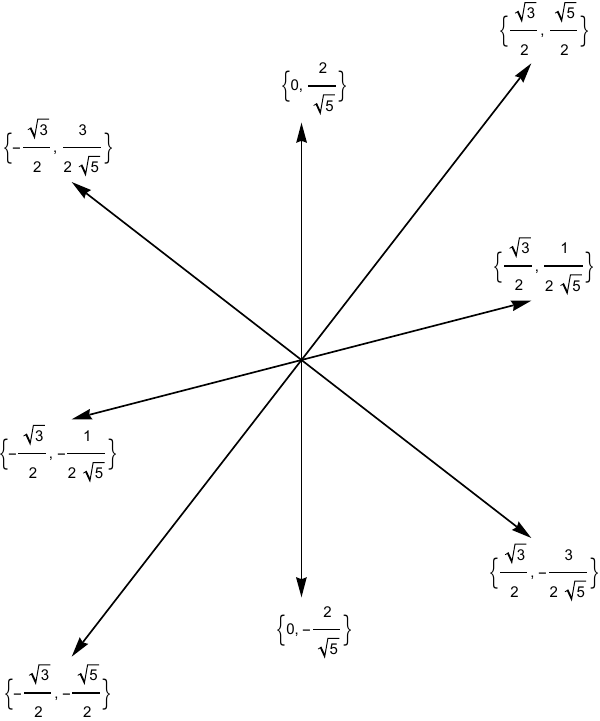}
\end{figure}
\noindent The multiplicities are given by
\begin{equation*}
    \widehat{c}_{\frac{2\widetilde{e}_2}{\sqrt{5}}} = \widehat{c}_{\frac{\sqrt{3}\widetilde{e}_1}{2} + \frac{\widetilde{e}_2}{2\sqrt{5}}} = 2, \quad 
    \widehat{c}_{\frac{\sqrt{3}\widetilde{e}_1}{2} + \frac{\sqrt{5}\widetilde{e}_2}{2}} = \frac15, \quad
    \widehat{c}_{\frac{\sqrt{3}\widetilde{e}_1}{2} - \frac{3 \widetilde{e}_2}{2 \sqrt{5}}} = 1.
\end{equation*}
These vectors are as in a deformed $C_2$ configuration 
except for having different multiplicities.

By formula~\eqref{simpleform}, the corresponding generalised trigonometric spin Calogero--Moser Hamiltonian is
\begin{equation*}
\begin{array}{c}
\displaystyle{
    L = \partial_{y_1}^2 + \partial_{y_2}^2 - 
    \frac{3 + 3 \widehat{P}_{\frac{\sqrt{15} \widetilde{e}_1 - 3 \widetilde{e}_2}{2 \sqrt{5}}}}{10 \sinh^2(\frac{\sqrt{15}y_1 - 3 y_2}{4 \sqrt{5}})} -
    \frac{4 + 2 \widehat{P}_{\frac{2 \widetilde{e}_2}{\sqrt{5}}}}{5 \sinh^2(\frac{y_2}{\sqrt{5}})}
    }
    \\ 
    \displaystyle{
     - \frac{4 + 2 \widehat{P}_{\frac{\sqrt{15}\widetilde{e}_1 + \widetilde{e}_2}{2 \sqrt{5}}}}{5 \sinh^2 (\frac{\sqrt{15}y_1 + y_2}{4 \sqrt{5}})}
     - \frac{1 + 5 \widehat{P}_{\frac{\sqrt{3}\widetilde{e}_1 + \sqrt{5} \widetilde{e}_2}{2}}}{50 \sinh^2 (\frac{\sqrt{3}y_1 + \sqrt{5}y_2}{4})},
    }
\end{array}
\end{equation*}
where $y_1 = \frac{1}{\sqrt{3}}(x_8 - x_7 - x_6)$ and $y_2 = \frac{1}{\sqrt{5}}(x_1 + \dots + x_5)$.

The scalar version of this operator coincides with the trigonometric degeneration of the operator from \cite[Theorem 1.4]{Taniguchi} for $a=\sqrt{3/5}$.

Other choices of a subgraph $\Gamma_0 \simeq A_4$ in the Dynkin diagram of~$E_6$ lead to equivalent projected configurations.
All other allowed two-dimensional projections of $E_6$ are as in the table \cite[p.~273]{Feigin}. 

\begin{remark}
   With a view towards generalising the above two operators $L$ corresponding to the restrictions $(E_7,D_5)$ and $(E_6,A_4)$,
   a natural question is whether there exists a spin version of the one-parametric family of integrable operators from \cite[Theorem 1.4]{Taniguchi}. In the trigonometric limit, the  Hamiltonian should have the form
   \begin{align}\label{Taniguchi}
        L &= \partial_{x_1}^2 + \partial_{x_2}^2 - 
        \frac{(3-a^2)\left(3-a^2 + 4a^2 \widetilde{P}_{e_1}\right)}{4a^2 \sinh^2(2ax_1)} -
        \frac{(3a^2-1)\left(3a^2-1+4\widetilde{P}_{e_2}\right)}{4 \sinh^2(2x_2)}
         \\ 
         &\qquad - \frac{2(a^2+1)\left(2+\widetilde{P}_{ae_1 + e_2}\right)}{\sinh^2(ax_1+x_2)}
         - \frac{2(a^2+1)\left(2+\widetilde{P}_{-ae_1 + e_2}\right)}{\sinh^2(-ax_1+x_2)} \nonumber
    \end{align}
    for some matrices $\widetilde{P}$. Note that in the cases $a=\sqrt{2}$ and $a=\sqrt{3/5}$ seen above, the operator~\eqref{Taniguchi} has a degree $6$ quantum integral since the Weyl groups of type $E_7$ and $E_6$ have basic invariants of degree~$6$. Note also that the scalar version of the operator~\eqref{Taniguchi} has a quantum integral of degree $6$ for any $a$ by~\cite{Taniguchi}.  
\end{remark}

\hfill \newline
\noindent (v) $(E_7, A_1^3)$ \\
Let us choose a subgraph $A_1^3$ in the Dynkin diagram of $E_7$ as follows

\begin{equation*}
\begin{tikzpicture}
    \filldraw [black] (0,0) circle (3pt) node[below right, black] {$\alpha_4$} ;
    \filldraw [red] (-1.5,0) circle (3pt) node[below, black]  {$\alpha_3$};
    \filldraw [black] (-3,0) circle (3pt) node[below,black] {$\alpha_1$};
    \filldraw [red] (1.5,0) circle (3pt) node[below,black] {$\alpha_5$};
    \filldraw [black] (3,0) circle (3pt) node[below,black] {$\alpha_6$};
    \filldraw [red] (0,-1.5) circle (3pt) node[right,black] {$\alpha_2$};
    \filldraw [black] (4.5,0) circle (3pt) node[below,black] {$\alpha_7$};
    \draw[thick] (-3,0) -- (-1.6,0);
    \draw[thick][black] (-1.4,0) -- (1.4,0);
    \draw[thick][black] (1.6,0) -- (3,0);
    \draw[thick][black] (0,0) -- (0,-1.4);
    \draw[thick][black] (3.1,0) -- (4.5, 0);
\end{tikzpicture}
\end{equation*}
The corresponding subspace $\pi$ is defined by the equations
\begin{equation*}
    x_1 = x_2 = 0, \ x_3 = x_4, 
\end{equation*}
and we require that $c = \frac{1}{2}$. Expressed in terms of $\widetilde{e}_1 = \frac12( e_3 + e_4 )$, $\widetilde{e}_2 = \frac12 (e_5 + e_6)$, $\widetilde{e}_3 =\frac12 (e_6-e_5)$, and $\widetilde{e}_4 = \frac12 (e_8-e_7)$, the projection onto $\pi$ of the positive half of $E_7$ leads to vectors $2\widetilde{e}_i$ ($i=1, \dots, 4$), $\widetilde{e}_1$, $\widetilde{e}_2 \pm \widetilde{e}_3$, $\pm \widetilde{e}_2 + \widetilde{e}_4$, $\pm \widetilde{e}_3 + \widetilde{e}_4$, $\pm \widetilde{e}_1 + \widetilde{e}_2 \pm \widetilde{e}_3$, $\pm \widetilde{e}_1 \pm \widetilde{e}_2 + \widetilde{e}_4$, and $\pm \widetilde{e}_1 \pm \widetilde{e}_3 + \widetilde{e}_4$, with multiplicities $\frac12$, $4$, $2$, $2$, $2$, $1$, $1$, and $1$, respectively. The total number of vectors in the positive half of the projected system is $23$.  
This configuration corresponds to the entry $(E_7, A_1^3)_2$ in the table~\cite[p.~272]{Feigin} --- notice that the configuration becomes non-reduced in the trigonometric case.

\subsubsection{Type $F_4$}
\hfill \newline

\noindent (i) $(F_4 , A_1^2)$  \\
The Dynkin diagram of the root system $F_4$ is
\begin{equation*}
\begin{tikzpicture}
\filldraw[red]  (0,0) circle (3pt) node[below right, black] {$\alpha_3$} ;
\filldraw[black] (-1.5,0) circle (3pt) node[below, black]  {$\alpha_2$};
\filldraw[red] (-3,0) circle (3pt) node[below,black] {$\alpha_1$};
\filldraw[black] (1.5,0) circle (3pt) node[below,black] {$\alpha_4$};
\draw[thick][black] (-2.9,0) -- (-1.6,0);
\draw[>=stealth, thick, ->,double] (-1.4,0) -- (-0.1,0);
\draw[thick][black] (0.1,0) -- (1.5, 0);

\end{tikzpicture}
\end{equation*}
with the simple roots being 
\begin{equation*}
    \alpha_1 = e_2 - e_3, \ \alpha_2 = e_3 - e_4, \ \alpha_3 = e_4, \ \alpha_4 = \frac{1}{2}(e_1 - e_2 - e_3 - e_4),
\end{equation*}
and here we chose a subgraph $\Gamma_0  \simeq A_1^2$. The corresponding plane~$\pi$ is given by the equations
\begin{equation*}
    x_2 = x_3, \ x_4 = 0,
\end{equation*}
with multiplicities being $c_1 = c_2 = \frac{1}{2}$, where $c_1 = c_{\alpha_1}$ and $c_2 = c_{\alpha_3}$. The corresponding projected system is shown in the following diagram, where the coordinates are in the basis formed by $e_1$ and $\widetilde{e}_2 = \frac{1}{\sqrt{2}}(e_2 + e_3)$:
\begin{figure}[H]
\includegraphics[scale=0.48]{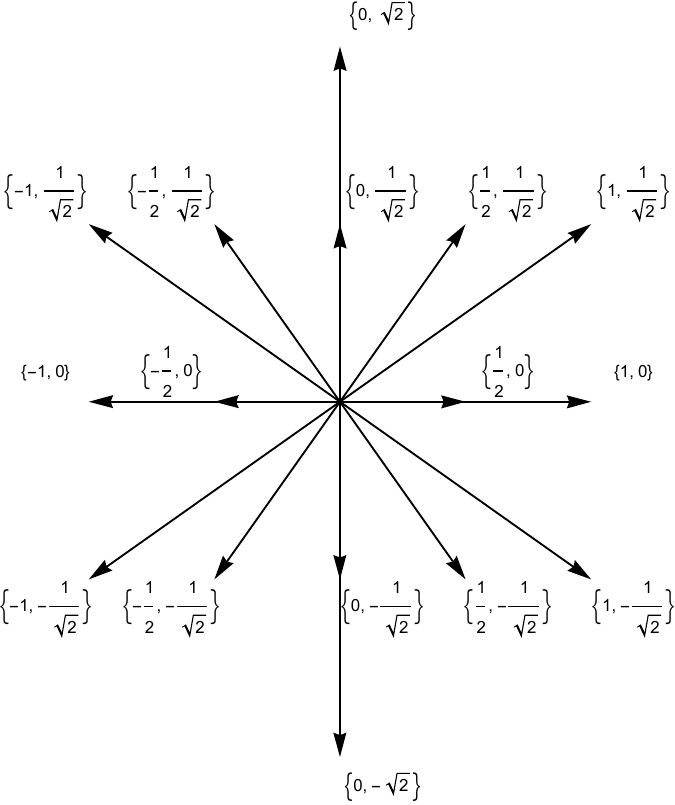}
\end{figure}
\noindent The multiplicities are given by
\begin{equation*}
    \widehat{c}_{\frac{e_1}{2}} = 2, \quad \widehat{c}_{e_1} = \frac32, \quad
    \widehat{c}_{\frac{\widetilde{e}_2}{\sqrt{2}}} = 3, \quad
    \widehat{c}_{\sqrt{2}\widetilde{e}_2} = \frac12, \quad
    \widehat{c}_{e_1 \pm \frac{\widetilde{e}_2}{\sqrt{2}}} = \widehat{c}_{\frac{e_1}{2} \pm \frac{\widetilde{e}_2}{\sqrt{2}}} = 1.
\end{equation*}

By formula~\eqref{simpleform}, the corresponding generalised trigonometric spin Calogero--Moser Hamiltonian is
\begin{align*}
    L &= \partial_{y_1}^2 + \partial_{y_2}^2 - \frac{5 + \widehat{P}_{ \frac{e_1}{2}}}{8 \sinh^2(\frac{y_1}{4})} - 
    \frac{9 + 6 \widehat{P}_{e_1}}{16 \sinh^2(\frac{y_1}{2})}
    \\
    & \quad - 
    \frac{12 + 3 \widehat{P}_{\frac{\widetilde{e}_2}{\sqrt{2}}}}{8 \sinh^2 (\frac{y_2}{2 \sqrt{2}})} -
    \frac{1 + 2 \widehat{P}_{\sqrt{2}\widetilde{e}_2}}{8 \sinh^2 (\frac{y_2}{\sqrt{2}})} - 
    \frac{3 + 3 \widehat{P}_{\frac{\sqrt{2} e_1 + \widetilde{e}_2}{\sqrt{2}}}}{8 \sinh^2(\frac{y_1}{2}+\frac{y_2}{2 \sqrt{2}})} 
    \\
    &\quad -
   \frac{3 + 3 \widehat{P}_{\frac{\sqrt{2} e_1 - \widetilde{e}_2}{\sqrt{2}}}}{8 \sinh^2(\frac{y_1}{2}-\frac{y_2}{2 \sqrt{2}})}  -
    \frac{3 + 3 \widehat{P}_{\frac{e_1 + \sqrt{2} \widetilde{e}_2}{2}}}{16 \sinh^2(\frac{y_1}{4}+\frac{y_2}{2 \sqrt{2}})} -
    \frac{3 + 3 \widehat{P}_{\frac{e_1 - \sqrt{2}\widetilde{e}_2}{2}}}{16 \sinh^2(\frac{y_1}{4}-\frac{y_2}{2 \sqrt{2}})},
\end{align*}
where $y_1 = x_1$ and $y_2 = \frac{1}{\sqrt{2}}(x_2 + x_3)$. This operator is a trigonometric version with spin of the potential-gauge form of the operator~\cite[Formula (28)]{Feigin} for $m = \frac{7}{2}$, $n=0$, and $\alpha = \sqrt{2}$.

Other choices of a subgraph $\Gamma_0 \simeq A_1^2$ in the Dynkin diagram of~$F_4$ lead to equivalent projected configurations.

\hfill \newline
\noindent (ii) $(F_4, A_2)$ \\
As another example, let us choose a subgraph $A_2$ in the Dynkin diagram of $F_4$ as in the following picture 
\begin{equation*}
\begin{tikzpicture}
\filldraw[black]  (0,0) circle (3pt) node[below right, black] {$\alpha_3$} ;
\filldraw[red] (-1.5,0) circle (3pt) node[below, black]  {$\alpha_2$};
\filldraw[red] (-3,0) circle (3pt) node[below,black] {$\alpha_1$};
\filldraw[black] (1.5,0) circle (3pt) node[below,black] {$\alpha_4$};
\draw[thick][red] (-3,0) -- (-1.5,0);
\draw[>=stealth, thick, ->,double] (-1.4,0) -- (-0.1,0);
\draw[thick][black] (0,0) -- (1.5, 0);

\end{tikzpicture}
\end{equation*}
with multiplicity $c_1 = \frac{1}{3}$ and $c_2$ being a free parameter.
The corresponding plane $\pi$ is defined by the equations
\begin{equation*}
    x_2 = x_3 = x_4. 
\end{equation*}
The corresponding projected system is shown in the following diagram, where the coordinates are in the basis $e_1$ and $\widetilde{e}_2 = \frac{1}{\sqrt{3}}(e_2 + e_3 + e_4)$:
\begin{figure}[H]
\includegraphics[scale=0.5]{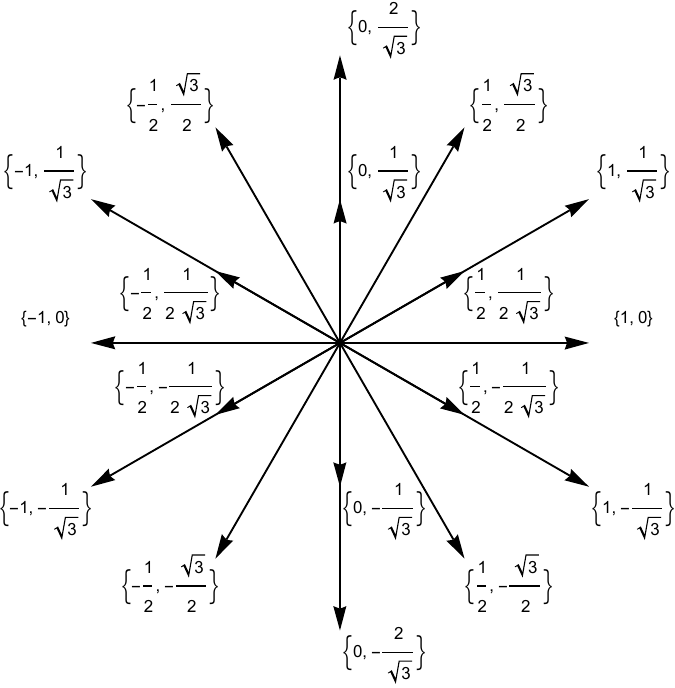}
\end{figure}
\noindent These vectors and their multiplicities coincide with those of the configuration called~$AG_2$.

By formula~\eqref{simpleform}, the corresponding generalised trigonometric spin Calogero--Moser Hamiltonian is
\begin{align*}
    L &= \partial_{y_1}^2 + \partial_{y_2}^2 - \frac{c_2(c_2 + \widehat{P}_{e_1})}{4 \sinh^2(\frac{y_1}{2})} 
         - \frac{c_2\left(c_2 + \widehat{P}_{\frac{e_1}{2} + \frac{\sqrt{3}\widetilde{e}_2}{2}}\right)}{4 \sinh^2(\frac{y_1}{4} + \frac{\sqrt{3}y_2}{4})} - \frac{c_2\left(c_2 + \widehat{P}_{\frac{e_1}{2} - \frac{\sqrt{3}\widetilde{e}_2}{2}}\right)}{4 \sinh^2(\frac{y_1}{4} - \frac{\sqrt{3}y_2}{4})} \\
       &- \frac{c_2\left(3c_2 + 2 + \widehat{P}_{\frac{\widetilde{e}_2}{\sqrt{3}}}\right)}{4 \sinh^2(\frac{y_2}{2\sqrt{3}})} 
       - \frac{c_2\left(3c_2 + 2 + \widehat{P}_{\frac{e_1}{2}+\frac{\widetilde{e}_2}{2\sqrt{3}}}\right)}{4 \sinh^2 (\frac{y_1}{4} + \frac{y_2}{4 \sqrt{3}})} 
       - \frac{c_2\left(3c_2 + 2 + \widehat{P}_{\frac{e_1}{2}-\frac{\widetilde{e}_2}{2\sqrt{3}}}\right)}{4 \sinh^2 (\frac{y_1}{4} - \frac{y_2}{4 \sqrt{3}})} \\
       &- \frac{1+\widehat{P}_{\frac{2 \widetilde{e}_2}{\sqrt{3}}}}{3 \sinh^2(\frac{y_2}{\sqrt{3}})} 
       -  \frac{1 + \widehat{P}_{e_1 + \frac{\widetilde{e}_2}{\sqrt{3}}}}{3 \sinh^2(\frac{y_1}{2} + \frac{y_2}{2\sqrt{3}})} 
       - \frac{1 + \widehat{P}_{e_1 - \frac{\widetilde{e}_2}{\sqrt{3}}}}{3 \sinh^2(\frac{y_1}{2} - \frac{y_2}{2\sqrt{3}})},
\end{align*}
where $y_1 = x_1$ and $y_2 = \frac{1}{\sqrt{3}}(x_2 + x_3 + x_4)$. The scalar version of this operator 
reproduces (up to rescaling and rotating the configuration of vectors) the operator \cite[Formula (1.1)]{FV19} with $m = c_2$.

The other possible choice of a subgraph $\Gamma_0 \simeq A_2$ in the Dynkin diagram of~$F_4$ leads to $G_2$ with multiplicities $3c_1 + 1$ and $c_1$ for the short and long roots, respectively. The subgraph $\Gamma_0 \simeq B_2$ leads to $BC_2$ with multiplicities $4c_2$, $4c_1 + c_2$, and $c_1$ for the vectors $e_i$, $e_i \pm e_j$, and $2e_i$, respectively, where $2(c_1 + c_2)=1$.

\subsubsection{Type $G_2$}
One of the two one-dimensional restrictions of the scalar trigonometric $G_2$ Calogero--Moser operator gives the operator
\begin{align}
\label{op3}
    L &= \frac{d^2}{dx^2} - \frac{8 m_1 (m_1 + 1)}{\sinh^2 x}-\frac{4 m_1 (m_1 + 1)}{\sinh^2 2x}-\frac{18}{\sinh^2 3x} \\ \ \nonumber \\
     &= \frac{d^2}{dx^2} - \frac{(3m_1+1)(3m_1+2)}{\sinh^2 x} - \frac{m_1(m_1+1)}{\sinh^2(x-\frac{\boldsymbol{\pi} \boldsymbol{i}}{2})} \nonumber \\
     &\qquad - \frac{2}{\sinh^2(x-\frac{\boldsymbol{\pi} \boldsymbol{i}}{3})} 
     - \frac{2}{\sinh^2(x-\frac{2\boldsymbol{\pi} \boldsymbol{i}}{3})} \nonumber
\end{align}
with $m_1 \in \C$.

The operator~\eqref{op3} is monodromy-free for all $m_1 \in \Z$ (see~\cite{DG} for local conditions defining monodromy-free potentials). 
Moreover, it is easy to check that the operator \eqref{op3} is the only monodromy-free operator of the form 
\begin{equation*}
    L = \frac{d^2}{dx^2} - \sum\limits_{m = 1}^n 
    \frac{c_m}{\sinh^2 m x}
\end{equation*}
with $n=3$ and $c_3 \neq 0$. 

\subsection{Projections that give root systems}

Table~\ref{tbl} lists all the non-trivial cases where the projected configuration is a root system. For the exceptional root systems, we only list their projections of rank at least~$2$. 
In each case, the list of multiplicities is ordered by the length of the vectors (in increasing order).
As previously, we denote the multiplicity of the roots $e_i \pm e_j$ in type $B$, $C$, and $F$ by $c_1$, and we denote by $c_2$ the multiplicity of $e_i$ in type $B$ and $F$, and the multiplicity of $2e_i$ in type~$C$.

\begin{center}
\begin{table}
    \begin{longtable}[c]{|l | c | c|} 
         \hline
         $(\Gamma, \Gamma_0)$ & Projection & Multiplicities  \\ 
         \hline & & \\
         $(A_{mk-1}, A_{k-1}^m)$, \ \ $m, k \in \Z_{\geq 2}$ & $A_{m-1}$ & $k$  \\ & & \\
         $(B_{mk+l}, A_{k-1}^m\times B_l)$, \ \ $m, k \in \Z_{\geq 1}$, $l \in \Z_{\geq 0}$ & $BC_m$ & $(2c_1 l + c_2) k$, \ $c_1k^2$, \ $\frac{k-1}{2}$  \\ 
         & & If $l>0$, $c_2 = \frac{1}{2} - c_1(l - 1)$, \\
         & & and if $k>1$, $c_1 = \frac1k$.  \\ & & \\
         $(C_{mk+l}, A_{k-1}^m \times C_l)$, \ \ $m, k \in \Z_{\geq 1}$, $l \in \Z_{\geq 0}$ & $BC_m$ & $2c_1 kl$, \ $c_1k^2$, \ $\frac{k-1}{2} + c_2k$ \\ 
         & & If $l>0$, $c_2 = \frac{1}{2} - c_1(l - 1)$, \\
         & & and if $k>1$, $c_1 = \frac1k$.  \\ & & \\
         $(D_{mk+l}, A_{k-1}^m \times D_l)$, \ \ $m \in \Z_{\geq 1}$, $k \in \Z_{\geq 2}$,  & $BC_m$ & $2l$, \, $k$, \ $\frac{k-1}{2}$ \\
         \hspace{3.2em} $l=0$, or $k$ is even and $l=\frac{k}{2} + 1$  & &  \\ & & \\
         $(D_{m+l}, D_l)$, \ \ $m \in \Z_{\geq 1}$, $l \in \Z_{\geq 2}$ & $B_m$ & $\frac{l}{l-1}$, \ $\frac{1}{2(l-1)}$   \\ & & \\
         $(F_4, A_2)$ \ with $\Gamma_0^v = \{\alpha_3, \alpha_4 \}$ & $G_2$ & $3c_1 + 1$, \ $c_1$ \\ & & \\
         $(F_4, B_2)$ & $BC_2$ & $4c_2$, \ $4c_1 + c_2$, \ $c_1$ \\
          & & with $2(c_1 + c_2)=1$ \\ & & \\
         $(E_6,A_2^2)$ & $G_2$ & $3$, \ $\frac13$ \\ & & \\
         $(E_6,D_4)$ & $A_2$ & $\frac43$ \\ & & \\
         $(E_7, A_1^3)$ \ with $\Gamma_0^v = \{ \alpha_2, \alpha_5, \alpha_7 \}$ & $F_4$ & $2$, \ $\frac12$  \\ & & \\ 
         $(E_7, D_4)$ & $C_3$ & $\frac43$, \ $\frac16$ \\ & & \\
         $(E_7, A_5)$ \ with $\Gamma_0^v = \{\alpha_2, \alpha_4, \alpha_5, \alpha_6, \alpha_7\}$ & $G_2$ & $\frac52$, \ $\frac16$ \\ & & \\
         $(E_8, D_4)$ & $F_4$ & $\frac43$, \ $\frac16$ \\ & & \\
         $(E_8, D_6)$ & $BC_2$ & $\frac{16}{5}$, \ $\frac{6}{5}$, \ $\frac{1}{10}$ \\ & & \\
         $(E_8, E_6)$ & $G_2$ & $\frac{27}{12}$, \ $\frac{1}{12}$
         \\ [1ex] 
         \hline
    \end{longtable}
    \setcounter{table}{0}
    \caption{Pairs $(\Gamma, \Gamma_0)$ such that the projection is a root system.}
    \label{tbl}
    \end{table}
\end{center}
\newpage

By using formula~\eqref{simpleform}, we get Hamiltonians for each of the projected root systems with multiplicities as given in Table~\ref{tbl} and with the matrices $\widehat{P}_\alpha$ defined in terms of a representation of the Weyl group corresponding to $\Gamma$.

\begin{remark}
In the rational case, there is additionally the stratum $(H_4, I_2(5))$ leading to a projected root system of type $I_2(10)$ and the corresponding Hamiltonian from Theorem~\ref{potential gauge}.
\end{remark}

Let us stress that even though the singularities of the potential for the restricted Hamiltonians corresponding to the Table~\ref{tbl} are given by the 
Weyl
group mirrors, the numerators of the corresponding matrix potentials are in general different from the standard matrix 
Hamiltonians as considered by Cherednik in \cite{Ch2}. The construction of restricted Hamiltonians produces a richer class of operators. Indeed, the input data to construct our potential in an $n$-dimensional space includes a representation of a Weyl group of higher rank $N \ge n$. The standard Hamiltonians from \cite{Ch2} are recovered in the case when $N=n$ and the two Weyl groups are the same. In this case the subgraph $\Gamma_0$ is empty and the projection of the root system coincides with itself; such cases are not listed in Table~\ref{tbl}. Example \ref{B_2 example} above gives an illustration of a rational matrix potential related to the configuration of singularities $B_2$ which is  obtained by a restriction from the $B_3$ case and which is different from $B_2$ Cherednik's Hamiltonians.

\subsection{The case of the root system $BC_N$} \label{BC}
Trigonometric Dunkl operators for the root system $BC_N$ are defined by formula~\eqref{trigDunkl} with $R_+ = BC_N^+ = \{e_i, \, 2e_i \mid i=1, \dots, N\} \cup \{e_i \pm e_j \mid 1 \leq i < j \leq N \}$ and multiplicities $c_{e_i \pm e_j} = c_1$, $c_{e_i} = c_2$, and $c_{2e_i} = c_3$. These Dunkl operators commute but are not equivariant; however, $\mathcal{B}_N$-invariant polynomials of $\nabla_i^{\mathrm tr} = \nabla_{e_i}^{\mathrm tr}$ commute with $W = \mathcal{B}_N$.

Let $P = \bigoplus_{i=1}^N \Z e_i$ be the weight lattice of the root system $BC_N$.
The trigonometric Cherednik algebra of type~$BC_N$ can be defined by its faithful polynomial representation on $\C[\{e^{(\alpha , x)} \mid \alpha \in P \}]$ and is generated by~$W$,~$e^{(\alpha , x)}$ ($\alpha \in P$), and $\nabla_i^{\mathrm tr}$.

In the case of $BC_N$, it is possible to define the following affine version of invariant parabolic strata. Let us consider the plane $\pi^a \subset V = \C^N$ defined by the equations
\begin{equation}    
\begin{aligned}
\label{BCstrata}
    &x_1 = x_2 = \dots = x_k, \ x_{k+1} = x_{k + 2} = \dots = x_{2k}, \ \dots, \\
    &x_{(m-1)k+1} = x_{(m-1)k+2} = \dots = x_{mk}, \\
    &x_{N - l_1 + 1} = x_{N - l_1 + 2} = \dots = x_{N} = 0, \\
    &x_{N-l_1-l_2 +1} = \dots = x_{N-l_1} = \boldsymbol{\pi} \boldsymbol{i}
\end{aligned}
\end{equation}
for $l_1, l_2 \in \Z_{\geq 0}$ and $m, k \in \Z_{\geq 1}$ such that $mk + l_1 + l_2 \leq N$.
Define the affine parabolic stratum $D^a = \cup_{w \in W} w(\pi^a)$. 
Let $x_0 \in \pi^a$ be generic, and let $\mathcal{C}^p_{Wx_0}(V, U) \subset \mathcal{C}_{Wx_0}(V, U)$ be the subspace of those (formal) sums of germs $f \in \mathcal{C}_{Wx_0}(V, U)$ that are periodic in the sense that if $x, x' \in Wx_0$ satisfy $x - x' \in 2\boldsymbol{\pi} \boldsymbol{i} P$ then 
$f(y + x'-x) = f(y)$ for all $y \in V$ in a small neighbourhood of~$x$.
The natural action of $W$ and $e^{(\alpha,x)}$ ($\alpha \in P$) on $\mathcal{C}_{Wx_0}(V, U)$ preserves~$\mathcal{C}^p_{Wx_0}(V, U)$, and
the Dunkl operators $\nabla_i^{\mathrm tr}$ also act on $\mathcal{C}^p_{Wx_0}(V, U)$.
The next theorem gives the conditions under which
the subspace $\mathcal{I}^p \subset \mathcal{C}^p_{Wx_0}(V, U)$ consisting of those elements that vanish when restricted to $D^a$ is preserved by this action of the $BC_N$ trigonometric Cherednik algebra.

\begin{theorem} 
    The space $\mathcal{I}^p$ is invariant under the $BC_N$ trigonometric Cherednik algebra if and only if all of the following conditions hold:
    \begin{align}
    \label{cond}
        &c_1 = \frac1k \text{ if } k>1, \quad 2(c_2+c_3) + 2(l_1-1)c_1 = 1 \text{ if } l_1>0, \text{ and } \\
        & \nonumber 2c_3 + 2(l_2-1)c_1 = 1 \text{ if } l_2>0.
    \end{align}
\end{theorem}
    The proof is essentially the same as for other root systems with additional use of the periodicity property, thus we omit it.

This result allows us to act on the module $ \mathcal{C}^p_{Wx_0}(V, U) / \mathcal{I}^p$. 
We have
\begin{equation*}
\begin{array}{c}
\displaystyle{
    \sum\limits_{i = 1}^N (\nabla^{\rm tr}_i)^2 = \sum\limits_{i = 1}^N \partial_{x_i}^2 - \sum\limits_{i = 1}^N c_2 \coth(\tfrac{x_i}{2}) \partial_{x_i} - \sum\limits_{i = 1}^N 2 c_3 \coth(x_i) \partial_{x_i} 
    }
    \\ 
\displaystyle{
    - \sum\limits_{i < j}^N c_1 \left( \coth(\tfrac{x_i - x_j}{2}) (\partial_{x_i} - \partial_{x_j}) + \coth(\tfrac{x_i + x_j}{2}) (\partial_{x_i} + \partial_{x_j}) \right)
    }
    \\  
\displaystyle{
    + \sum\limits_{i = 1}^N \frac{c_2}{4 \sinh^2(\frac{x_i}{2})} (1 - s_{e_i}) + \sum\limits_{i = 1}^N \frac{c_3}{\sinh^2(x_i)} (1 - s_{e_i}) 
    }
    \\ 
\displaystyle{
    + \sum\limits_{i < j}^N \left( \frac{c_1}{2 \sinh^2(\frac{x_i - x_j}{2})} (1 - s_{e_i - e_j}) + \frac{c_1}{2 \sinh^2(\frac{x_i + x_j}{2})} (1 - s_{e_i + e_j}) \right) + (\rho, \rho ),
    }
    \end{array}
\end{equation*}
in the scalar case, where $\rho = \frac12\sum_{\alpha \in BC_N^+} c_\alpha \alpha$.
After the restriction and the gauge transformation to the potential gauge by the formal function
\begin{equation*}
    f = \prod_{\substack{\widehat\alpha \in \widehat{R}_+, t_{\alpha} \\ (\widehat{\alpha}, t_{\alpha}) \neq (0,0)}}\left( \sinh\left(\frac{(\widehat{\alpha}, y) + t_{\alpha}}{2}\right)\right)^{\widehat{c}_{\widehat{\alpha}, t_{\alpha}}} 
\end{equation*}
we obtain the following integrable matrix operator, similarly to the case of non-affine parabolic strata:
\begin{equation}
\label{BCL}
    L = \Delta_y - \sum\limits_{\widehat{\alpha} \in \widehat{R}_+, \, t_{\alpha}} \frac{\widehat{c}_{\widehat{\alpha},t_{\alpha}}(\widehat{c}_{\widehat{\alpha}, t_{\alpha}} + 2 \widehat{c}_{2\widehat{\alpha}, 2t_{\alpha}} + \widehat{P}_{\widehat{\alpha}, t_{\alpha}}) (\widehat{\alpha}, \widehat{\alpha})}{4\sinh^2(\frac{(\widehat{\alpha},y) + t_{\alpha}}{2})}, 
\end{equation}
where the sum now goes over pairs $(\widehat{\alpha}, t_{\alpha}) \neq (0,0)$, where $\widehat{\alpha}$ is the orthogonal projection of the root $\alpha$ onto the plane $\pi$ which is parallel to $\pi^a$ and passes through $0$,  $t_{\alpha}\in \mathbb C$ is a shift that arises due to the affine nature of the plane $\pi^a$ (see below), and $y \in \pi$. Let us note that we consider $t_{\alpha}$ to be defined modulo $2 \boldsymbol{\pi} \boldsymbol{i}$, thus we identify $(\widehat{\alpha}, 2\boldsymbol{\pi} \boldsymbol{i})$ and $( \widehat{\alpha}, 0)$. The formula for $\widehat{P}_{\widehat{\alpha}, t_{\alpha}}$ is almost the same as in the non-affine case:
\begin{equation*}
    \widehat{P}_{\widehat{\alpha}, t_{\alpha}} = 1 + \frac{1}{\widehat{c}_{\widehat{\alpha}, t_{\alpha}} (\widehat{\alpha}, \widehat{\alpha})} \sum_{\substack{ \gamma \in R_+ \\(\widehat{\gamma}, t_{\gamma}) = (\widehat{\alpha}, t_{\alpha})}} 
    c_{\gamma}(\gamma, \gamma) (P_{\gamma} - 1),
\end{equation*}
where
\begin{equation*}
    \widehat{c}_{\widehat{\alpha}, t_{\alpha}} = \sum\limits_{\substack{\gamma \in R_+  \\ (\widehat{\gamma}, t_{\gamma}) = (\widehat{\alpha}, t_{\alpha})}} c_{\gamma}.
\end{equation*}
The parameters $t_\alpha$, $\alpha \in R$ are defined as $t_\alpha=(\alpha, \delta)$, where $\delta =  \boldsymbol{\pi} \boldsymbol{i} \sum_{j=N - l_1 - l_2}^{N- l_1 - 1}  e_j$ gives the parallel translation between the planes $\pi$ and $\pi^a$. 
Explicitly, the list of the projected roots $\widehat{\alpha} \neq 0$ and shift parameters~$t_{\alpha}$ is
\begin{align*}
    &\{ (\widetilde{e}_i \pm \widetilde{e}_j, 0) \mid m + 1 \leq i < j \leq m+n \} \\ &\cup \{ (\tfrac{\widetilde{e}_i}{\sqrt{k}} \pm \widetilde{e}_j , 0) \mid 1 \leq i \leq m, \, m+1 \leq j \leq m+n \} \\ &\cup \{ (\tfrac{\widetilde{e}_i \pm \widetilde{e}_j}{\sqrt{k}} , 0) \mid 1 \leq i < j \leq m\}  \\ &\cup \{(\widetilde{e}_i ,0 ) \mid m+1 \leq i \leq m+n \}  \cup \{ (2\widetilde{e}_i , 0) \mid m+1 \leq i \leq m+n \} \\  &\cup \{ (\tfrac{\widetilde{e}_i}{\sqrt{k}} , 0) \mid 1 \leq i \leq m \} \cup \{ (\tfrac{2\widetilde{e}_i}{\sqrt{k}},0) \mid 1 \leq i \leq m \} \\
    &\cup \{ (\widetilde{e}_i,  \boldsymbol{\pi} \boldsymbol{i}) \mid m+1 \leq i \leq m+n  \} \cup \{ (\tfrac{\widetilde{e}_i}{\sqrt{k}},  \boldsymbol{\pi} \boldsymbol{i}) \mid 1 \leq i \leq m \}. 
\end{align*}
 The respective multiplicities are
\begin{align*}
    &\widehat{c}_{\widetilde{e}_i \pm \widetilde{e}_j, 0} = c_1, \quad \widehat{c}_{\tfrac{\widetilde{e}_i}{\sqrt{k}} \pm \widetilde{e}_j , 0} = c_1 k, \quad \widehat{c}_{\tfrac{\widetilde{e}_i \pm \widetilde{e}_j}{\sqrt{k}} , 0} = c_1 k^2,  \quad
    \widehat{c}_{\widetilde{e}_i, 0} = c_2 +2 c_1 l_1, \\
    &\widehat{c}_{2\widetilde{e}_i, 0} = c_3, \quad 
    \widehat{c}_{\tfrac{\widehat{e}_i}{\sqrt{k}}, 0} = c_2 k + 2 c_1 k l_1, \quad
    \widehat{c}_{\tfrac{2\widehat{e}_i}{\sqrt{k}}, 0} = c_3 k + \frac{k-1}{2}, \\
    &\widehat{c}_{\widetilde{e}_i, \boldsymbol{\pi} \boldsymbol{i}} = 2 c_1l_2, \quad
    \widehat{c}_{\tfrac{\widetilde{e}_i}{\sqrt{k}},  \boldsymbol{\pi} \boldsymbol{i}} = 2 c_1 kl_2.
\end{align*}
Let us note that since all shifts are either 0 or $ \boldsymbol{\pi} \boldsymbol{i}$, it is, in principle, possible to rewrite the potential using only $\sinh$ functions without shifts.

The derivation of formula~\eqref{BCL} is very similar to the computations done in the non-affine case in Theorem~\ref{trigpot}, but in several steps we 
need to use the explicit form of the strata \eqref{BCstrata} and particular properties of the root system $BC$. First of all, we need the following lemma.
\begin{lemma}
    The following identity holds
    \begin{equation}
    \label{form2}
    \sum_{\substack{\alpha, \beta  \in R_+ \\  (\widehat{\beta}, t_{\beta}) \nsim (\widehat{\alpha}, t_{\alpha})}} c_{\alpha} c_{\beta} (\widehat{\alpha}, \widehat{\beta}) \coth\left(\frac{(\widehat{\alpha}, y) + t_{\alpha}}{2}\right) \coth\left(\frac{(\widehat{\beta}, y) + t_{\beta}}{2}\right) = \rm{const},
    \end{equation}
    where $\sim$ is defined by $(\widehat{\alpha}, t_{\alpha}) \sim (\widehat{\beta}, t_{\beta})$ iff $\exists \; r \in \mathbb{C}: \widehat{\alpha} = r \widehat{\beta}$ and $t_{\alpha} - r t_{\beta} \in 2 \boldsymbol{\pi} \boldsymbol{i} \mathbb{Z}$.
\end{lemma}
\begin{proof}
    Let us give a sketch of the proof. Similarly to the proof of the Theorem \ref{trigpot} we note that for the $BC$ case the condition for not being proportional ensures the expression \eqref{form2} only has simple poles. When computing the residue on the hyperplanes $(\widehat{\alpha}, y ) + t_{\alpha} \in 2 \boldsymbol{\pi} \boldsymbol{i} \mathbb{Z}$ we get the following identity
    \begin{equation*}
    \sum\limits_{\substack{\beta \in R_+ \\ (\widehat{\beta}, t_{\beta}) \nsim (\widehat{\alpha}, t_{\alpha})}} 
    c_{\beta} (\widehat{\alpha},\widehat{\beta}) \coth\left(\frac{(\widehat{\beta}, y) + t_{\beta}}{2}\right) = 0,
    \end{equation*}
    for $(\widehat{\alpha}, y ) + t_{\alpha} \in 2 \boldsymbol{\pi} \boldsymbol{i} \mathbb{Z}$, which is proved by decomposition of this sum into the orbits as in the proof of Theorem \ref{trigpot}, where we also need to use that all non-zero constants $t_{\gamma} = \boldsymbol{\pi} \boldsymbol{i}$ for $\widehat{\gamma} \neq 0$ and $(e_i^{\vee}, \gamma) \in 2  \mathbb{Z}$, $\forall \gamma \in BC_N$. The absence of poles for the expression \eqref{form2} forces it to be constant, which does not provide any non-trivial contribution to the Hamiltonian. 
\end{proof}
\begin{remark}
    In the case of the reflection representation, the operators $\widehat{P}_{\widehat{\alpha}, t_{\alpha}}$ become reflections on the plane $\pi$ about the hyperplane orthogonal to the projected root $\widehat{\alpha}$, and in particular, do not depend on the shift parameter $t_{\alpha}$.
\end{remark}

\begin{remark}
    Similarly to the Table \ref{tbl}, in the $BC_N$ case with the corresponding strata \eqref{BCstrata} the projected configuration is a root system in the case when $N = m k + l_1 + l_2$, which gives the root system $BC_m$.
\end{remark}

\section{Extra integrals for deformed spin Calogero--Moser system}\label{extra integals}
Since we are dealing with matrix operators, we anticipate the existence of more quantum integrals than just the ones coming from invariant combinations of Dunkl operators. In this section, we explain how additional integrals 
can be obtained for the deformed spin Calogero--Moser system from Section~\ref{trig type A} in the case of $U = (\mathbb{C}^r)^{\otimes N}$ ($r \in \Z_{>0}$) by using
the Drinfeld functor and any maximal commutative subalgebra inside a Yangian.

We will use the trigonometric Dunkl operators~\eqref{trigDunkl} for $R_+ = A_{N-1}^+$ and $\xi = e_i$. More precisely, we let
\begin{equation*}
    \nabla_i^{\rm tr} =  \partial_{x_i} + \sum\limits_{j < i} \frac{c}{1 - e^{x_i - x_j}} (1 - s_{ij}) - \sum\limits_{j > i} \frac{c}{1 - e^{x_j - x_i}} (1 - s_{ij}) - c (i-1).
\end{equation*}
We will also need a second version of trigonometric Dunkl operators, introduced by Polychronakos~\cite{Polych} (cf. Heckman's operators \cite{HeckmanTrig}), given by
\begin{equation*}
D_i = \partial_{x_i} - \sum\limits_{\substack{ j=1 \\ j \neq i}}^N \frac{c}{1 - e^{x_j - x_i}} (1 - s_{ij}) = \nabla_i^{\rm tr} + c \sum\limits_{j < i} s_{ij},
\end{equation*}
which do not commute, but satisfy the relations
\begin{equation*}
\begin{array}{c}
     \displaystyle{
     s_{ij} D_i = D_j s_{ij},
     }
     \\ \ \\
     \displaystyle{
     [D_i, D_j] = c (D_i - D_j)s_{ij}.
     }
\end{array}    
\end{equation*}
\subsection{Drinfeld functor}
Suppose that $U = (\mathbb{C}^r)^{\otimes N}$ for some $r \in \Z_{>0}$. Then one can introduce the following 
formal power series with values in ${\rm End}(\widetilde{M})$: 
\begin{equation}
\label{Yangian}
    t^{ab}(u) = \delta_{ab} - \sum\limits_{i = 1}^N \frac{c E^{ab}_i}{u - D_i}, \qquad 1 \leq a,b \leq r,
\end{equation}
where $E^{ab}$ is the matrix unit with $1$ in the $(a,b)$-th entry and zeroes everywhere else, and the subscript indicates on which tensor component this matrix acts non-trivially.
This defines a representation of the Yangian $Y(gl_r)$ \cite{D, A, BAB}, meaning that the matrix $t(u)$ satisfies the following $RTT$
relation:
\begin{equation*}
    R_{12}(u - v) t_1(u) t_2(v) = t_2(v) t_1(u) R_{12}(u - v)
\end{equation*}
with
\begin{equation*}
    R_{12}(u-v) = 1 - \frac{c}{u - v}\sum\limits_{a,b = 1}^r E_1^{ab}\otimes E_2^{ba}
\end{equation*}
and
\begin{equation*}
    \begin{array}{ccc}
    t(u) = \sum\limits_{a,b = 1}^r t^{ab}(u) E^{ab}, & t_1(u) = t(u) \otimes 1, & t_2(u) = 1 \otimes t(u).
    \end{array}
\end{equation*}
\subsection{Gelfand--Tsetlin subalgebra and quantum integrals}
It is well known that the coefficients of the expansion of the quantum determinant
\begin{equation*}
    {\rm  det}_q(u) = \sum\limits_{\sigma \in S_r} (-1)^{\sigma} t^{1, \sigma(1)}(u-c(r-1)) t^{2, \sigma(2)}(u-c(r-2)) \cdots 
    t^{r, \sigma(r)} (u)
\end{equation*}
generate the centre of the Yangian. Moreover, the representation \eqref{Yangian} of the Yangian preserves the subspace
$\widetilde{M}^{S_N}\subset \widetilde{M}$. 

The quantum determinant can be computed explicitly in this representation to be \cite{BAB}:
\begin{equation} \label{qdet}
    {\rm det}_q(u) = \frac{\prod\limits_{i = 1}^N (u  - \nabla^{\rm tr}_i - c)}{\prod\limits_{i = 1}^N (u - \nabla^{\rm tr}_i)}.
\end{equation}
Thus, the Hamiltonians of the deformed spin Calogero--Moser system are contained in the image of the centre $\mathcal{Z}(Y(gl_r))$ of the Yangian in the representation $\widetilde{M}^{S_N}$. 

Let us recall that
inside the Yangian, there exists a large commutative subalgebra called the Gelfand--Tsetlin subalgebra. It is
described in terms of a tower of embeddings of Yangians.
\begin{equation*}
    Y(gl_1) \subset Y(gl_2) \subset \cdots \subset Y(gl_r).
\end{equation*}
Namely, it is the subalgebra
\begin{equation*}
    A = \langle \mathcal{Z}(Y(gl_1)), \mathcal{Z}(Y(gl_2)) , \ldots, \mathcal{Z}(Y(gl_r)) \rangle.
\end{equation*}
This algebra is commutative, and its image in the representation $\widetilde{M}^{S_N}$ contains the Hamiltonians of the deformed spin Calogero--Moser system, which means it provides a large number of additional quantum integrals for this model, which can be summarised in the following result.
\begin{theorem}
    The image of the centre of the Yangian \eqref{qdet} in the representation
    $\widetilde{M}^{S_N}$ is generated by
    the integrals \eqref{trighigherham} of the generalised trigonometric spin
    Calogero--Moser system \eqref{trigham}.
    The image of the Gelfand--Tsetlin subalgebra $A$ or any other maximal commutative 
    subalgebra in the Yangian is a commutative subalgebra containing all the 
    integrals \eqref{trighigherham}. Moreover,
    the image of the Yangian in the representation $\widetilde{M}^{S_N}$ 
    commutes with the operators \eqref{trighigherham}.
\end{theorem}

For the (undeformed) spin Calogero--Moser system, a joint eigenbasis for the image of the Gelfand--Tsetlin subalgebra $A$ was constructed in \cite{TakemuraUglov}. In the case of $r=2$ (that is, for $gl_2$) and the Bethe subalgebra of the Yangian, a joint eigenbasis was constructed recently in \cite{Ferrando}. 

\begin{example}
     The simplest additional integral, obtained from the 
     expansion of $t_{11}(u)$,  for the non-deformed trigonometric spin 
     Calogero--Moser system is
    \begin{equation*}
        h = \sum\limits_{i = 1}^N E^{11}_i \partial_{x_i} - 
        \sum\limits_{j \neq i}^N E^{11}_i \frac{c}{1 - e^{x_j - x_i}}(1 - P_{ij}).
    \end{equation*}
    In the deformed case, for $N = 3$, $m = 1$, $n = 1$, and 
    $c = \frac{1}{2}$, this integral becomes
    \begin{align*}
        h^{\rm def}  = \frac{1}{2}(E^{11}_1 + E^{11}_2) \partial_y &+ E^{11}_3
        \partial_{x} -  \frac{1}{2(1 - e^{x - \frac{y}{\sqrt{2}}})} 
        (E^{11}_1 - E^{11}_1 P_{13} + E^{11}_2 - E_2^{11} P_{23})  \\ 
        &- \frac{1}{2(1 - e^{\frac{y}{\sqrt{2}} - x})} E^{11}_3 
        (2 - P_{13} - P_{23}).
    \end{align*}
\end{example}

\begin{remark}
    A similar construction exists for the realisation of the spin Calogero--Moser Hamiltonian of the $BC_N$ type together with its $\mathcal{B}_N$-invariant integrals as central elements of the reflection algebra \cite{CC}. Commutative subalgebras of the reflection algebra then produce other quantum integrals.
\end{remark}

\subsection{Extra integrals for the rational model}
For the rational deformed spin Calogero--Moser model, additional quantum integrals 
can be obtained using the current algebra $u^{-1} gl_r [u^{-1}]$, which can be 
viewed as the associated graded algebra of the filtered Yangian $Y(gl_r)$. 
Namely, let $\nabla_i$ be the Dunkl operators \eqref{Dunkl} for $gl_N$, 
then we can define a representation of the current algebra $u^{-1} gl_r [u^{-1}]$ 
on the space  $(\mathbb{C}^r)^{\otimes N} \otimes M$.
Similarly to the Yangian case, the generators $e^{ab} \otimes u^{-l}$ ($l \leq -1$) 
of the algebra $u^{-1} gl_r [u^{-1}]$ act by the formula
\begin{equation*}
\label{current}
    \sum\limits_{i = 1}^N E^{ab}_i \nabla_i^{l - 1}.
\end{equation*}

It is straightforward to check that the space of 
diagonal invariants $((\mathbb{C}^r)^{\otimes N} \otimes M)^{S_N}$ is 
preserved by the action of $u^{-1} gl_r [u^{-1}]$.
The generators of the centre of the current algebra act by the elements
\begin{equation} \label{currentcentre}
    C^0_l = \sum\limits_{a = 1}^r \sum\limits_{i = 1}^N E^{aa}_i \nabla_i^{l - 1} = \sum\limits_{i = 1}^N \nabla_i^{l - 1},
\end{equation}
which in full analogy with the Yangian case provides us with rational 
deformed spin Calogero--Moser Hamiltonians, while the generators of 
the Gelfand--Tsetlin type subalgebra act by the elements 
\begin{equation}
    \label{GToper}
    C^m_l = \sum\limits_{a = 1}^{r - m} \sum\limits_{i = 1}^N E^{aa}_i \nabla_i^{l - 1}, \quad 0 \leq m \leq r - 1,
\end{equation}
which gives us additional quantum integrals for this model. 
Similarly to the trigonometric case we can summarise this 
section in the following theorem.
\begin{theorem}
    The image of the centre of the current algebra \eqref{currentcentre} in the
    representation
    $\widetilde{M}^{S_N}$ is generated by
    the higher integrals \eqref{higherham} of the generalised rational spin
    Calogero--Moser system \eqref{H2}.
    The image of the Gelfand--Tsetlin subalgebra or any other maximal commutative 
    subalgebra in the current algebra is a commutative algebra containing the 
    integrals \eqref{higherham}. Moreover,
    the image of the current algebra in the representation $\widetilde{M}^{S_N}$ 
    commutes with operators \eqref{higherham}.
\end{theorem}

\begin{example}
    The simplest additional integral for the non-deformed rational spin Calogero--Moser system can be obtained either by considering the Gelfand--Tsetlin subalgebra
    in the current algebra or by degenerating the integral from the trigonometric
    case, which gives
    \begin{equation*}
        h = \sum\limits_{i = 1}^N E^{11}_i \partial_{x_i} - \sum\limits_{j \neq
        i}^N E^{11}_i \frac{c}{x_i - x_j} (1 - P_{ij}).
    \end{equation*}
    In the deformed case, for $N = 3$, $m = 1$, $n = 1$, and $c = 
    \frac{1}{2}$, this integral  becomes
     \begin{align*}
        h^{\rm def}  = \frac{1}{2}(E^{11}_1 + E^{11}_2) \partial_y &+ E^{11}_3
        \partial_{x} -   \frac{1}{\sqrt{2} y - 2x} 
        (E^{11}_1 - E^{11}_1 P_{13} + E^{11}_2 - E_2^{11} P_{23})  \\ 
        &- \frac{1}{ 2x - \sqrt{2}y} E^{11}_3 
        (2 - P_{13} - P_{23}).
    \end{align*}
\end{example}

\begin{remark}
Another way to obtain additional integrals for various rational Hamiltonians which we consider is to apply the restriction procedure of Section \ref{spin case} and a generalisation of Theorem \ref{maintheorem} to elements of the centraliser ${\mathcal H}_c^W$ which commute with $\sum_{i=1}^N \nabla_i^2$ and do not have the form $p(\nabla)$ for a polynomial $p$. A large family of such elements in ${\mathcal H}_c^W$ can be given in terms of Dunkl Laplace--Runge--Lenz vector and Dunkl angular momenta (see \cite{feiginhakobyan}).

\end{remark}

\subsection*{Acknowledgements}
We are very grateful to C.\,Korff for many joint discussions in the early stages of this work.
We would also like to thank   V.\,Caudrelier, O.\,Chalykh, N.\,Cramp\'e, and J.\,Lamers for useful comments.

The work of M.\,Feigin was supported by the Engineering and Physical Sciences Research Council [grant number  EP/W013053/1]. 
The work of M.~Vrabec was funded by a Carnegie--Caledonian PhD scholarship from the Carnegie Trust for the Universities of Scotland. M.~Vasilev is grateful to the University of Glasgow for financial support.

\appendix

\section{An explicit higher integral}\label{B_2 integral}

Recall the operator $L$ from Example~\ref{B_2 example}, which is a two-dimensional restriction of a Hamiltonian for the root system $B_3$.
The $4$th-order integral $L_4 = \widetilde{\Res}_{\pi}( \sum_{i = 1}^3 \nabla_i^4)$ of the operator $L$ 
for $c_1 = 1$ is the matrix differential operator $L_4 = (L_4^{(i,j)})_{i,j=1}^2$ with 
\begin{align*}
    L_4^{(1,1)} &= 
    \partial_{x_1}^4 + \partial_{x_2}^4
    -\frac{19}{2 x_1^2} \partial_{x_1}^2
    -\frac{35}{2x_2^2} \partial_{x_2}^2
    +\frac{3 x_1^2-19 x_2^2}{x_1^3 \left(x_1^2-x_2^2\right)}\partial_{x_1}
    +\frac{35 x_1^2-51 x_2^2}{x_2^3\left(x_1^2-x_2^2\right)}\partial_{x_2} \\ \ \\
    &\quad +\frac{385 x_1^6+511 x_1^4 x_2^2 -479 x_1^2 x_2^4 -161 x_2^6}{16 x_1^4 x_2^4\left(x_1^2-x_2^2\right)},
\end{align*}
\begin{align*}
    L_4^{(2,1)} &= 
    \frac{8}{x_1^2}\partial_{x_1}^2-\frac{8}{x_2^2} \partial_{x_2}^2
    -\frac{16 \left(3x_1^2-x_2^2\right)}{x_1^3 \left(x_1^2-x_2^2\right)}\partial_{x_1} +\frac{16 \left(x_1^2-3 x_2^2\right)}{x_2^3\left(x_1^2-x_2^2\right)}\partial_{x_2} \\ \ \\
    &\quad 
   +\frac{2 \left(7 x_1^6+ x_1^4x_2^2+x_1^2x_2^4+7 x_2^6\right)}{x_1^4x_2^4 \left(x_1^2-x_2^2\right)},
\end{align*}
\begin{align*}
    L_4^{(1,2)} &=
    -\frac{8 \left(3 x_1^2 x_2^2-x_2^4\right)}{x_1^2 \left(x_1^2-x_2^2\right){}^2}\partial_{x_1}^2
    -\frac{8 \left(x_1^2+x_2^2\right)}{\left(x_1^2-x_2^2\right){}^2}\partial_{x_2}^2
    -\frac{16 x_1 x_2}{\left(x_1^2-x_2^2\right){}^2}\partial_{x_1}\partial_{x_2} \\ \ \\
    &\quad +\frac{8 \left(3 x_1^6+ x_1^4 x_2^2 -2 x_1^2 x_2^4 +2 x_2^6\right)}{x_1^3\left(x_1^2-x_2^2\right){}^3}\partial_{x_1}
    +\frac{8 \left(4 x_1^4-11 x_1^2 x_2^2 +3 x_2^4\right)}{x_2\left(x_1^2-x_2^2\right){}^3}\partial_{x_2} \\ \ \\
    &\quad -\frac{2 \left(21 x_1^{10}-87 x_1^8 x_2^2 +166 x_1^6 x_2^4 -20 x_1^4 x_2^6 -23 x_1^2 x_2^8 +7 x_2^{10}\right)}{x_1^4 x_2^2 \left(x_1^2-x_2^2\right){}^4},
\end{align*}
and

\begin{align*}
    L_4^{(2,2)} &= \partial_{x_1}^4+\partial_{x_2}^4 
    -\frac{51 x_1^4-6 x_1^2 x_2^2 +19 x_2^4}{2 x_1^2 \left(x_1^2-x_2^2\right){}^2}\partial_{x_1}^2
    -\frac{3 x_1^4+26 x_1^2x_2^2 +35x_2^4}{2x_2^2 \left(x_1^2-x_2^2\right){}^2 }\partial_{x_2}^2 
    -\frac{32 x_1 x_2 }{\left(x_1^2-x_2^2\right){}^2}\partial_{x_1}\partial_{x_2} \\ \ \\
    &\quad  
    +\frac{51 x_1^6-41 x_1^4 x_2^2 +73 x_1^2x_2^4 -19x_2^6}{x_1^3 \left(x_1^2-x_2^2\right){}^3 }\partial_{x_1} 
    +\frac{3 x_1^6+7 x_1^4x_2^2 -71 x_1^2x_2^4 -3 x_2^6}{x_2^3\left(x_1^2-x_2^2\right){}^3}\partial_{x_2} \\ \ \\
    &\quad -\frac{63 x_1^{12}-60  x_1^{10}x_2^2-1959 x_1^8x_2^4 +8200 x_1^6x_2^6 -2567 x_1^4x_2^8 +580 x_1^2x_2^{10} -161 x_2^{12}}{16 x_1^4x_2^4 \left(x_1^2-x_2^2\right){}^4  }.
\end{align*}

\begin{small}

\end{small}
 

\begin{thebibliography}{99}

\footnotesize{

\bibitem{A}
T.\ Arakawa, \textit{Drinfeld functor and finite-dimensional representations of Yangian}, Comm.\ Math.\ Phys., Volume 205, (1999) pages~1--18.

\bibitem{AHL}
F.\ Atai, M.\ Hallnäs, and E.\ Langmann, \textit{Orthogonality of super-Jack polynomials and a Hilbert space interpretation of deformed Calogero--Moser--Sutherland operators}, Bull.\ Lond.\ Math.\ Soc., Volume 51, Issue 2, (2019) pages 353--370.

\bibitem{BC}
Y.\ Berest and O.\ Chalykh, \textit{Deformed Calogero--Moser operators and ideals of rational Cherednik algebras}, Comm.\ Math.\ Phys., Volume 400, (2023) pages 133--178.

\bibitem{BL}
Y.\,Y.\ Berest and I.\,M.\ Loutsenko,  \textit{Huygens' principle in Minkowski spaces and soliton solutions of the Korteweg-de Vries equation}, Comm.\ Math.\ Phys., Volume 190, (1997) pages 113–132. 

\bibitem{BAB}
D.\ Bernard, M.\ Gaudin, F.\,D.\,M.\ Haldane, and V.\ Pasquier, \textit{Yang--Baxter equation in spin chains with long range interactions}, J.\ Phys.\ A: Math.\ Gen., Volume 26, Issue 20,  (1993) pages 5219--5236.

\bibitem{BAB1}
E.\ Billey, J.\ Avan, and O.\ Babelon, \textit{Exact Yangian symmetry in the classical Euler--Calogero--Moser model}, Phys.\ Lett.\ A, Volume 188, Issue 3, (1994) pages~263--271.


\bibitem{Calogero}
F.\ Calogero, \textit{Solution of the one-dimensional $N$-body problem with quadratic and/or inversely quadratic pair potential}, J.\ Math.\ Phys., Volume 12, (1971) pages~419--436.

\bibitem{CC}
V.\ Caudrelier and N.\ Cramp\'e, \textit{Integrable $N$-particle Hamiltonians with Yangian or reflection algebra symmetry}, 
J.\ Phys.\ A, Volume 37, (2004) pages 6285--6298.

\bibitem{CFV'98}
O.\,A.\ Chalykh, M.\,V.\ Feigin, and A.\,P.\ Veselov, \textit{New integrable generalizations of Calogero–Moser quantum problem}, J.\ Math.\ Phys., Volume 39, (1998) pages~695--703.

\bibitem{CFV'99}
O.\,A.\ Chalykh, M.\,V.\ Feigin, and A.\,P.\ Veselov, \textit{Multidimensional Baker--Akhiezer functions and Huygens’ principle}, Comm.\ Math.\ Phys., Volume 206, (1999) pages~533--566.

\bibitem{CGV}
O.\,A.\ Chalykh, V.\,M.\ Goncharenko, and A.\,P.\ Veselov, \textit{Multidimensional integrable Schr\"odinger operators with matrix potential}, J.\ Math.\ Phys., Volume 40, Issue 11, (1999) pages~5341--5355.

\bibitem{ChS}
O.\,A.\ Chalykh and A.\ Silantyev, \textit{
KP hierarchy for the cyclic quiver}, J.\ Math.\ Phys., Volume 58, Issue 7, (2017) 071702.


\bibitem{CVlocus}
O.\,A.\ Chalykh and A.\,P.\ Veselov, \textit{Locus configurations and $\vee$-systems}, Phys.\ Lett.\ A, Volume 285, Issue 5, (2001) pages~339--349.

\bibitem{Kimura}
H-Y.\ Chen, T.\ Kimura, and N.\ Lee, \textit{Quantum integrable systems from supergroup gauge theories}, J. High Energy Phys., Volume 2020, (2020) 104.

\bibitem{Ch1}
I.\ Cherednik, {\em A unification of Knizhnik--Zamolodchikov and Dunkl operators via affine Hecke algebras}, Invent.\ Math., Volume 106, (1991) pages~411--431.

\bibitem{Ch2}
I.\ Cherednik, {\em Integration of quantum many-body problems by affine Knizhnik--Zamolodchikov equations}, Adv.\ Math., Volume 106, (1994) pages~65--95.

\bibitem{Crampe}
N.\ Cramp\'e, {\em New spin generalisation for long range interaction models}, Lett. Math. Phys., Volume 77, (2006) pages 127--137. 

\bibitem{CrampeYoung}
N.\ Cramp\'e and C.\,A.\,S.\ Young, {\em  
Sutherland models for complex reflection groups},
Nuclear Physics B, Volume 797, Issue 3, (2008) pages 499--519. 



\bibitem{D}
V.\,G.\ Drinfeld, \textit{Degenerate affine Hecke algebras and Yangians}, Funct.\ Anal.\ Appl., Volume 20, Issue 1, (1986) pages~62--64.

\bibitem{DG} 
J.\,J.\ Duistermaat and F.\,A.\ Gr\"{u}nbaum, \textit{Differential equations in the spectral parameter}, Comm.\ Math.\ Phys., Volume 103, Issue 2, (1986) pages 177--240.

\bibitem{Dunkl} 
C.\,F.\ Dunkl, \textit{Differential-difference operators associated to reflection groups}, Trans.\ AMS, Volume 311, (1989) pages~167--183.

\bibitem{Et}
P.\ Etingof, \textit{Calogero--Moser systems and representation theory},  Zurich Lectures in Advanced Mathematics. European Mathematical Society (EMS), Z\"urich, 2007.


\bibitem{EG} 
P.\ Etingof and V.\ Ginzburg, \textit{Symplectic reflection algebras, Calogero--Moser space, and deformed Harish-Chandra homomorphism}, Invent.\ Math., Volume 147, (2002) pages~243--348. 

\bibitem{FF} 
A.\ Fairley and M.\ Feigin, \textit{Trigonometric planar real locus configurations} (in preparation).

\bibitem{FH}
M.\  Fairon and T. G\"orbe, \textit{Superintegrability of Calogero--Moser systems associated with the cyclic quiver}, Nonlinearity, Volume 34, Issue 11, (2021) pages 7662--7682. 


\bibitem{Feigin} 
M.\ Feigin, \textit{Generalized Calogero--Moser systems from rational Cherednik algebras},  Selecta Math.\ (N.S.), Volume 18, Issue 1, (2012) pages~253--281.

\bibitem{feiginhakobyan}
M.\ Feigin and T.\ Hakobyan, 
\textit{Algebra of Dunkl Laplace–Runge–Lenz vector}, Lett.\ Math.\ Phys., Volume 112, 59 (2022).

\bibitem{Feigin1}
M.\,V.\ Feigin, M.\,A.\ Hallnäs, and A.\,P.\ Veselov, \textit{Quasi-invariant Hermite polynomials and Lassalle--Nekrasov correspondence}, Comm.\ Math.\ Phys., Volume 386, (2021) pages 107--141. 

\bibitem{FS}
M.\ Feigin and A.\ Silantyev, \textit{Generalized elliptic Calogero--Moser systems} (in preparation).

\bibitem{FV19}
M.\ Feigin and M.\ Vrabec, \textit{Intertwining operator for $AG_2$ Calogero--Moser--Sutherland system}, J.\ Math.\ Phys., Volume 60, Issue 7, (2019) 073503.

\bibitem{Ferrando}
G.\ Ferrando, J.\ Lamers, F.\ Levkovich-Maslyuk, and D.\ Serban, \textit{Bethe Ansatz inside Calogero--Sutherland models}, SciPost Phys., Volume 18, 035 (2025).

\bibitem{GR}
D.\ Gaiotto and M.\ Rapčák, \textit{Miura operators, degenerate fields and the M2-M5 intersection},  J.\ High Energy Phys., Volume 2022, Issue 1, (2022) 35.

\bibitem{GH} 
J.\ Gibbons and T.\ Hermsen, \textit{A generalization of the Calogero-Moser system}, Physica D, Volume 11, Issue 3, (1984) pages 337--348.

\bibitem{HH}
Z.\ Ha and F.\ Haldane, \textit{Models with inverse-square exchange},
Phys.\ Rev.\ B, Volume 46, (1992) pages 9359--9368.

\bibitem{Hallnas}
M.\ Hallnäs, \textit{New orthogonality relations for super-Jack polynomials and an associated Lassalle--Nekrasov correspondence}, Constr.\ Approx., Volume 59, (2024) pages 113--142.

\bibitem{HL}
M.\ Hallnäs and E.\ Langmann, \textit{A unified construction of generalised classical polynomials associated with operators of Calogero--Sutherland type}, Constr.\ Approx., Volume 31, (2010) pages 309--342.


\bibitem{Heckman}
G.\,J.\ Heckman, \textit{A remark on the Dunkl differential-difference operators}, Prog.\ Math., Volume 101, (1991) pages~181--191.

\bibitem{HeckmanTrig}
G.\,J.\ Heckman, \textit{An elementary approach to the hypergeometric shift operators of Opdam}, Invent.\ Math., Volume 103, (1991) pages 341--350.

\bibitem{HeckmanSurvey}
G.\,J.\ Heckman, \textit{Dunkl operators}, Ast\'erisque, Volume 245, (1997) pages 223--246.

\bibitem{HW}
K.\ Hikami and M.\ Wadati, \textit{Integrability of Calogero--Moser spin system}, J.\ Phys.\ Soc.\ Jpn., Volume 62, Issue 2, (1993) pages 469--472.

\bibitem{Krich}
I.M.\  Krichever, \textit{Rational solutions of Kadomtsev--Petviashvili equation and integrable systems
of $N$ particles on a line}, Funct. Anal. Appl, Volume 12, Issue~1, (1978) pages 59--61.

\bibitem{KBBT} 
 I.\ Krichever, O.\ Babelon, E.\ Billey, and M.\ Talon, \textit{Spin generalization of the Calogero--Moser system and the matrix KP equation}, Amer.\ Math.\ Soc.\ Transl., Volume 170, Issue 2, (1995) pages 83--119.

\bibitem{LX}
L-C. \ Li and P.\ Xu, \textit{A class of integrable spin Calogero--Moser systems}, Comm. Math. Phys., Volume 231, (2002) pages 257--286.

\bibitem{LRS}
A.\ Liashyk, N.\ Reshetikhin, and I.\ Sechin, \textit{Quantum integrable systems on a classical integrable background}, Comm. Math. Phys., Volume 407, 15 (2026).

\bibitem{MP}
J.\,A. Minahan and A.\,P. Polychronakos, \textit{Integrable systems for particles with internal degrees of freedom}, Phys.\ Lett.\ B, Volume 302, Issues 2–3, (1993) pages 265--270.

\bibitem{Moser}
J.\ Moser, \textit{Three integrable Hamiltonian systems connected with isospectral deformation}, Adv.\ Math., Volume 16, (1975) pages~197--220.

\bibitem{OP}
M.\,A.\ Olshanetsky and A.\,M.\ Perelomov, \textit{Quantum completely integrable systems connected with semi-simple Lie algebras}, Lett.\ Math.\ Phys., Volume 2, (1977) pages~7--13.

\bibitem{Polych}
A.\,P.\ Polychronakos, \textit{Exchange operator formalism for integrable systems of particles}, Phys.\ Rev.\ Lett., Volume 69, Issue 5, (1992) pages 703--705.

\bibitem{Resh}
N.\ Reshetikhin, \textit{ 
Degenerate integrability of the spin Calogero--Moser systems and the duality with the spin Ruijsenaars systems}, Lett. Math. Phys., Volume~63, Issue 1, (2003) pages 55--71.


\bibitem{Serganova}
V.\ Serganova, \textit{On generalizations of root systems}, Commun.\ Algebra, Volume~24, Issue 13, (1996) pages 4281--4299.

\bibitem{Sergeev}
A.\,N.\ Sergeev, \textit{Calogero operator and Lie superalgebras}, Theor.\ Math.\ Phys., Volume 131, (2002) pages 747--764. 

\bibitem{SV}
A.\,N.\ Sergeev and A.\,P.\ Veselov, \textit{Deformed quantum Calogero--Moser problems and Lie superalgebras}, Comm.\ Math.\ Phys., Volume 245, (2004) pages~249--278.

\bibitem{SV'05}
A.\,N.\ Sergeev and A.\,P.\ Veselov, \textit{Generalised discriminants, deformed Calogero--Moser--Sutherland operators and super-Jack polynomials}, Adv.\ Math., Volume 192, Issue 2, (2005) pages 341--375. 

\bibitem{SV'09}
A.\,N.\ Sergeev and A.\,P.\ Veselov, \textit{$BC_\infty$ Calogero--Moser operator and super Jacobi polynomials}, Adv.\ Math., Volume 222, Issue 5, (2009) pages 1687--1726.

\bibitem{Sutherland}
B.\ Sutherland, \textit{Exact results for a quantum many-body problem in one dimension. II}, Phys.\ Rev.\ A, Volume 5, (1972) pages~1372--1376.

\bibitem{TakemuraUglov}
K.\ Takemura and D.\ Uglov, \textit{The orthogonal eigenbasis and norms of eigen-
vectors in the spin Calogero--Sutherland model}, J. Phys. A, Volume 30, Issue 10  (1997) pages 3685--3717.

\bibitem{Taniguchi}
K.\ Taniguchi, \textit{Deformation of two body quantum Calogero--Moser--Sutherland models}, arXiv:math-ph/0607053.

\bibitem{VFC}
A.\,P.\ Veselov, M.\,V.\ Feigin, and O.\,A.\ Chalykh, \textit{New integrable deformations of the Calogero--Moser quantum problem}, Russ.\ Math.\ Surv., Volume 51, (1996) pages~573--574.



}
\end{thebibliography}
\end{document}